\documentclass[twocolumn]{autart}   

\usepackage{amsmath}
\usepackage{algorithm}
\usepackage{algpseudocode}
\usepackage{amssymb}
\usepackage{color}
\usepackage{enumerate}
\usepackage{wrapfig}
\usepackage{graphicx}           

\usepackage[dvips]{epsfig}    
\usepackage[round, sort]{natbib}
\usepackage{hyperref}

\usepackage{booktabs}
\usepackage{multirow}
\usepackage{enumitem}
\usepackage{subfigure}

\hypersetup{
	colorlinks = true,
	citecolor = cyan,
}

\newtheorem{theorem}{Theorem}
\newtheorem{remark}{Remark} %[theorem]
\newtheorem{assumption}{Assumption}

\newtheorem{problem}{Problem}

\allowdisplaybreaks
\begin{document}	
\begin{frontmatter}

	\title{Data-Driven Robust MPC for Unknown Nonlinear Systems via Set-Membership Learning} 
	
	\thanks[footnoteinfo]{The work was supported in part by the National Natural Science Foundation of China under Grants U23B2059,  and 62088101. Frank Allg\"{o}wer is thankful that his work was funded by Deutsche Forschungsgemeinschaft (DFG, German Research Foundation) under Germany’s Excellence Strategy-EXC 2075-390740016 and within grant AL 316/15-1-468094890. The authors thank the International Max Planck Research School
		for Intelligent Systems (IMPRS-IS) for supporting Yifan Xie.
		\emph{(Corresponding author: Gang Wang.)} \\
	 \emph{Email address:} weiyuzhou@bit.edu.cn (Y. Wei), wenjie.liu@ntu.edu.sg (W. Liu), yifan.xie@ist.uni-stuttgart.de (Y. Xie),  frank.allgower@ist.uni-stuttgart.de (F. Allg\"{o}wer), sunjian@bit.edu.cn (J. Sun), gangwang@bit.edu.cn (G. Wang).}

	\author[Bit]{Yuzhou Wei},
	\author[Ntu]{Wenjie Liu}, 
	\author[Stuttgart]{Yifan Xie},
    \author[Stuttgart]{Frank Allg\"{o}wer},
    \author[Bit]{Jian Sun},    
	\author[Bit]{Gang Wang}

	\address[Bit]{School of Automation, Beijing Institute of Technology, Beijing 100081, China}      
     \address[Ntu]{School of Electrical and Electronic Engineering, Nanyang Technological University, Singapore 639798} 
	\address[Stuttgart]{Institute for Systems Theory and Automatic Control, University of Stuttgart, Stuttgart 70550, Germany}             
	%\address{Email:~\{,~,~,~,\\ 
	%		,~sunjian@bit.edu.cn,~\}}
		\maketitle

	%%%%%%%%%%%%%%%%%%%%%%%%%%%%%%%%%%%%%%%%%%%%%%%%%%%%%%%%%%%%%%%%%%%%%%%%%%%%%%%
	\begin{abstract}
    Data-driven model predictive control (MPC) has become an attractive approach for controlling unknown systems, especially when data are corrupted by noise. However, most existing data-driven MPC methods focus on linear systems, and little attention has been given to nonlinear dynamics under disturbances. To fill this gap, we propose a robust data-driven min-max MPC scheme for unknown nonlinear systems with process disturbances. We represent the unknown nonlinear dynamics using vector fields built from a dictionary of basis functions, yielding an equivalent linear form with unknown matrices. These unknown matrices are characterized by a set-membership representation derived from noisy input-state data. Using this uncertainty description, we formulate a min-max MPC problem. Two online scenarios are studied: i) when state measurements are noise-free, and, ii) when they are corrupted by process disturbance. For each case, we derive a Lyapunov-based semidefinite program (SDP) to compute a stabilizing state-feedback controller. The resulting schemes are shown to guarantee recursive feasibility and either exponential or robust stability of the closed-loop system depending on whether there is process disturbance. Simulation studies on benchmark examples illustrate the effectiveness and competitive performance of the proposed approach compared to existing data-driven and model-based controllers.
	\end{abstract}


       \end{frontmatter}	
	%%%%%%%%%%%%%%%%%%%%%%%%%%%%%%%%%%%%%%%%%%%%%%%%%%%%%%%%%%%%%%%%%%%%%%%%%%%%%%%%
\section{Introduction}\label{sec.intro} 
Data-driven control refers to the paradigm of designing controllers directly from experimental data, bypassing explicit system identification steps \cite{2013Model,Wangyb,datadriven1,datadriven2,wllw}. This approach has attracted significant interest because it can simplify the design as well as stability analysis and has been successfully applied to various domains such as robotics, energy systems, and process control \cite{ral2023zhou,locomotion,powerCDC}.  In particular, data-driven model predictive control (MPC) has emerged as a powerful framework for handling unknown systems
\cite{liu2022data1}. Like classical MPC, data-driven MPC can enforce state and input constraints and provide guarantees on closed-loop stability and performance, but it does so without requiring an identified model. For example, by leveraging the \emph{Willems et al}’s fundamental lemma, one can synthesize optimal control policies directly from measured input–output trajectories (e.g., the DeePC approach \cite{willems2005note}), 
A multitude of recent works have also extended data-driven MPC to handle practical issues such as output measurement noise and network attacks \cite{liu2022data1,liu2022data2}. 

Despite the effectiveness of data-driven MPC in handling measurement noise, many real-world systems are also influenced by process disturbances (due to unmodeled dynamics, unknown inputs, or adversarial interference), which directly affect state evolution. Unlike measurement noise, process disturbances require a different design strategy. Motivated by classical min-max MPC for model-based dynamical systems
\cite{Bemporadmm2003,MPCLMI1996,Wan2003}, recent research has begun to explore data-driven min–max MPC schemes that use offline noisy trajectories and online measurements to achieve robustness \cite{XieAuto}. These methods aim to adapt the control input to minimize the worst-case cost over all system models consistent with the data, thereby ensuring robust performance and stabilization under uncertainty.

In contrast to the extensive literature on linear systems, results for nonlinear systems are much more limited. A common strategy is to exploit known structural forms of the nonlinearity \cite{koopman,volterra}. For instance, the Koopman operator framework lifts a nonlinear system into a higher-dimensional linear space, enabling data-driven control for certain bilinear systems \cite{koopman}.  Other methods include linearly parameterized models with known basis functions, Gaussian process models, local linear or polynomial approximations, and feedback linearization \cite{gaussianprocesses,linear,cancellation}. For example,  work \cite{gaussianprocesses} provided data-driven control laws for general polynomial systems by carefully bounding the Taylor remainder. Alternatively,  \cite{cancellation} represented the vector field using a dictionary of functions and synthesized a controller via a data-dependent SDP. However, most of these works focus on static or offline controllers. To our knowledge, the problem of designing online data-driven state-feedback controllers for general nonlinear systems that explicitly account for both offline training data and online process disturbances remains open.

In this paper, we extend data-driven min–max MPC to unknown nonlinear systems with disturbances. Our approach begins by lifting the nonlinear dynamics into a pseudo-linear form using a library of known basis functions, thereby preserving the nonlinear structure. Specifically, we assume the unknown vector field can be written as $g(x)=A_sZ(x)$ for some unknown matrix $A_s$ and known lifting $Z(x)$. We then use noisy offline input–state data to construct a set-membership representation of the unknown system matrices. Based on this uncertainty set, we formulate a nonlinear data-driven min–max MPC problem. We consider two scenarios: i) online state measurements are noise-free, and, ii) they are corrupted by unknown but bounded disturbances. For each scenario, we derive a data-driven SDP using Lyapunov conditions, whose solution yields a stabilizing state-feedback controller gain. The resulting receding-horizon scheme guarantees exponential or robust stability of the closed-loop data-driven control system under uncertainty.

Our work is inspired by \cite{XieAuto,YifanX,weiauto2}. However, compared with these results, the primary technical challenge addressed here lies in the fact that the data-based SDP must simultaneously certify stability and constraint satisfaction for both the linear and nonlinear components of the system. This leads to the construction of a new constraint, in which nonlinear effects are incorporated through additional quadratic terms. Such a requirement cannot be met by a direct extension of the linear quadratic inequalities used in \cite{XieAuto,YifanX,weiauto2}. As a consequence, the derived recursive feasibility, constraint satisfaction, and robust stability guarantees for the closed-loop system explicitly account for the coupled effects of linear dynamics and nonlinear components, which are absent in the purely linear setting.
In a nutshell, the contributions of this paper are summarized as follows.
\begin{itemize} 
	\item[c1)] A representation of the unknown nonlinear system matrices is constructed from noisy input–state data via a quadratic matrix inequality. Two scenarios are then considered. Firstly, we focus on the case with offline noisy data and online noise-free state measurements. Secondly, we extend the framework to explicitly account for process disturbances in both the offline data and the online state measurements. In each case, the data-driven min–max MPC problem with state constraints is reformulated as an SDP that yields a state-feedback control law.
	\item[c2)] For the case with noise-free online measurements, we show that the resulting controller exponentially stabilizes the system to the origin and satisfies the state constraint. In the presence of noisy online measurements, we establish that the closed-loop system under the proposed scheme is robustly stable and ensures the state constraint.
\end{itemize}

The remainder of the paper is organized as follows. Section \ref{sec.prosta} describes the system model and presents a data-driven characterization of the unknown dynamics. Section \ref{sec.measure} develops a data-driven min–max MPC scheme for the case with noise-free online state measurements and establishes closed-loop guarantees, including recursive feasibility, state constraint satisfaction, and exponential stability. Section \ref{sec.noisy} extends the framework to the case with noisy online measurements and establishes recursive feasibility, robust stability, and constraint satisfaction. Section \ref{sec.sim} presents simulation results, and Section \ref{sec.con} concludes the paper.

\emph{Notation:} Let $\mathbb{R}(\mathbb{R}_+)$ and $\mathbb{Z}$ denote the sets of (positive) real numbers and nonnegative integers, respectively. We write $\mathbb{Z}_{\left[ a,b \right]}$ for $\{a,a+1,\ldots, b\}$. Identity and zero matrices are $I$ and ${0}$ (whose dimensions are clear from context). For matrices $Q$, $Q \succ 0$ ($Q \succeq 0$) means that $Q$ is positive definite (semi-definite). For vectors $x$ and matrices $Q \succ 0$ of compatible dimensions, $\| x \|_Q = \sqrt{x^\top Q x}$. We use the shorthand $QR[*]^\top$ for $QRQ^\top$ when dimensions match. 

\section{System Representation}\label{sec.prosta}
We consider the discrete-time nonlinear system
\begin{equation}	\label{eq.system1}
	x_{k+1} = g(x_k) + B_su_k + d_k,\quad k=1,2,3,\cdots
\end{equation}
where
$x_k \in \mathbb{R}^{n_x}$ is the system state, $u_k \in \mathbb{R}^{n_u}$ is the control input, and $d_k \in \mathbb{R}^{n_x}$ is an unknown process disturbance. The vector field $g(x_k) \in \mathbb{R}^{n_x}$ and the input matrix $B_s \in \mathbb{R}^{n_x \times n_u}$ are both unknown. Our goal is to design an online controller that stabilizes the nonlinear system despite the uncertainties.

Without any prior knowledge about the vector field $g(x_k)$, control design would be extremely challenging or even impossible. We therefore impose the following assumption.
\begin{assumption}[Unknown model]
	We assume there exists a known lifting function $Z: \mathbb{R}^{n_x} \to \mathbb{R}^{n_z}$ such that 
	\begin{equation}
		g(x_k) = A_s Z(x_k)
	\end{equation} 
	for some unknown constant matrix $A_s \in \mathbb{R}^{n_x \times n_z}$. \label{Ass.system}
\end{assumption}

In practice, one can take 
\begin{equation}\label{eq.z}
	Z(x_k)=\left[ \begin{matrix}
		x_k  \\
		Q(x_k)  \\
	\end{matrix} \right]	,
\end{equation}
where $Q(x)\in\mathbb{R}^{n_z-n_x}$ is a vector of known nonlinear basis functions. Under this assumption, the system dynamics can be written as
\begin{equation}	\label{eq.system}
	x_{k+1} = A_sZ(x_k) + B_su_k + d_k
\end{equation}
with unknown $(A_s,B_s)$. To restrict the class of admissible nonlinearities, we further impose the following quadratic constraint.
\begin{assumption}[Nonlinearities]
	For all states $x_k$ and their corresponding lifted vectors $Q(x_k)$ in system \eqref{eq.system},
	the inequality
	\begin{align}\label{eq.WSG}
		\begin{bmatrix}
			{x}_k \\
			Q({x}_k)
		\end{bmatrix}^\top
		\begin{bmatrix}
			-W & S \\
			S^\top & G
		\end{bmatrix}
		\begin{bmatrix}
			{x}_k \\
			Q({x}_k)
		\end{bmatrix}
		\le 0
	\end{align}
	holds, where $W = W^\top \succeq 0$, $G = G^\top \succeq 0$, and $S$ are known matrices. \label{Ass.non}
\end{assumption}
\begin{assumption}[Bounded process disturbance]
	The disturbance $d_k \in  \mathbb{R}^{n_x}$ is unknown but bounded, i.e., there exists a known constant ${\varpi} \ge 0$ such that $\| d_k\|_2^2\le \varpi^2$ for all $k$.
	\label{Ass.noise}
\end{assumption}

\begin{remark}
Assumption~\ref{Ass.system} implies that the nonlinear dynamics lie in the span of the chosen basis functions, whose selection can be informed by prior knowledge of the system physics. Under this viewpoint, this assumption corresponds to a grey-box modeling framework, wherein the functional structure of the dynamics is known while the associated parameters are unknown. Such an assumption is well motivated when a sufficiently expressive function dictionary is available, which is often the case for mechanical and electrical systems. In these settings, the governing equations can typically be derived from first principles, whereas the exact parameter values are difficult to identify a priori. Similar assumptions have been used in recent data-driven control works
	\cite{cancellation,HuZJ}. 
	Furthermore, Assumption \ref{Ass.system} considers a lifting function that encompasses both linear and nonlinear components. The special case where $Z(x_k) = x_k$ recovers the standard formulation for linear systems, which has been extensively studied in the literature. In contrast, choosing $Z(x_k) = Q(x_k)$ corresponds to a purely nonlinear representation and leads to a simplified results as discussed in \cite{cancellation}. This purely nonlinear case will be further explained in Remark \ref{re.xq}.
\end{remark}
\begin{remark}
	In Assumption~\ref{Ass.non}, the quadratic constraint in \eqref{eq.WSG} characterizes a broad class of nonlinearities of practical relevance, including Lipschitz or norm-bounded nonlinearities, sector-bounded nonlinearities, strongly convex functions with Lipschitz continuous gradients, and certain classes of recurrent neural networks \cite{nqc}. The matrices $W$, $S$, and $G$ encode prior structural information about the nonlinear mapping and are selected based on available analytical knowledge or known bounds. For instance, in many mechanical and electrical systems, energy balance relations or passivity arguments naturally give rise to quadratic inequalities of the form \eqref{eq.WSG}, from which  $W$, $S$, and $G$ can be constructed explicitly. In the case of Lipschitz nonlinearities, one may select $W = \ell^{2} I$, $S = 0$, and $G = -I$, where $\ell$ denotes an upper bound on the Lipschitz constant. For nonlinearities with bounded partial derivatives, the matrices $W$, $S$, and $G$ can be derived directly from the corresponding derivative bounds \cite{nqc}. Such quadratic constraints are commonly used tools in data-driven nonlinear control frameworks; see, e.g., \cite{HuZJ,nqc}. Furthermore, alternative classes of nonlinearities beyond those represented by inequality \eqref{eq.WSG} have also been investigated; see, e.g., \cite{HuZJ}.	 
\end{remark}
\begin{remark}
In Assumption \ref{Ass.noise}, the process disturbance  is assumed to be unknown but bounded by known upper limits. The bound $\varpi$ can be obtained using methods such as sensor noise equivalent standard deviation, variance- or quantile-based analyses \cite{DEA}, and historical experience.
\end{remark}

To handle the unknown matrices $A_s$ and $B_s$, we perform offline experiments to collect noisy input–state data. Denote the collected inputs and states by
\begin{align}\label{eq.XU}
	\bar U &:= \begin{bmatrix}
		\bar u_0 & \cdots & \bar u_k & \cdots & \bar u_{N-1}
	\end{bmatrix} \in \mathbb{R}^{n_u \times N}, \nonumber \\
	\bar X &:= \begin{bmatrix}
		\bar x_0 & \cdots &\bar x_k & \cdots & \bar x_{N}
	\end{bmatrix} \in \mathbb{R}^{n_x \times (N+1)},
\end{align}
where $\bar x_k \in  \mathbb{R}^{n_x}$ and $\bar u_k \in  \mathbb{R}^{n_u}$  are the state and input generated from \eqref{eq.system} with an arbitrary initial state $\bar{x}_0$. 
Since the lifting function $Z$ is known, we could compute $Z(x_k)$ from $x_k$ and the function $Q(x_k)$ for $k\in \mathbb Z_{[0, N]}$. The corresponding data matrix can be constructed as
\begin{align}\label{eq.Z}
	\bar Z  :=\left[ \begin{matrix}
		\bar x_0 \!& \cdots \!& \bar x_k \!& \cdots \!& \bar x_{N-1}  \\
		Q(\bar x_0) \!& \cdots \!& Q(\bar x_k) \!& \cdots \!& Q(\bar x_{N-1})  \\
	\end{matrix} \right]   \in {\mathbb R}^{{n_z}  \times   N}.
\end{align}
Furthermore, let us denote the sequence of unknown disturbance
\begin{align}\label{eq.D}
	\bar D &:= \begin{bmatrix}
		\bar d_0 & \cdots & \bar d_k & \cdots & \bar d_{N-1}
	\end{bmatrix} \in \mathbb{R}^{n_x \times N}.
\end{align}

By the system equation \eqref{eq.system}, each data point satisfies
\begin{equation*} 
	\bar d_k = \bar x_{k+1} - A_s  Z(\bar x_k) - B_s \bar u_k
\end{equation*}
and by Assumption \ref{Ass.noise},  $\| \bar d_k \|_2^2 \le {\varpi}^2$ for all $k \in \mathbb{Z}_{[0,N]}$. This implies that any pair of matrices $(A,B) \in \Xi_k$  consistent with the data at time $k$ satisfy a certain quadratic matrix inequality. Defining $\mathcal {\bar Q}_k =\left[ \begin{matrix}
	I   &  \bar x_{k+1}  \\
	0   &  -  Z(\bar x_k)  \\
	0 &   -\bar u_k  \\
\end{matrix} \right] \! 	
\left[ \begin{matrix}
	\varpi^2 I   &  0  \\
	0   &    -I  \\
\end{matrix} \right]\left[ \begin{matrix}
I   &  \bar x_{k+1}  \\
0   &  -  Z(\bar x_k)  \\
0  &   -\bar u_k  \\
\end{matrix} \right]^\top$, the following set characterizes all consistent system matrices
\begin{equation}\label{eq.Qk}
\Xi_k  = \left\{ \begin{aligned}
	&	( A,B) :  
	\left[ \begin{matrix}
		I & A & B \\
	\end{matrix} \right]\mathcal {\bar Q}_k \left[ \begin{matrix}
		I & A & B \\
	\end{matrix} \right]^\top \succeq 0  \\
\end{aligned} \right\}
\end{equation}
for $k=0,1, \cdots, N-1$, which must contain the true $(A_s,B_s)$. Here, the set
$\Xi_k$ is characterized using the data $\bar x_{k+1}$, $Z(\bar x_k)$, and $\bar u_k$. Intersecting over all collected samples from $k=0$ to $k=N-1$ yields the  intersection of all consistency sets $\Xi_k$
\begin{equation}\label{eq.Qkk}
	\mathcal{S} := \bigcap_{k=0}^{N-1} \Xi_k.
\end{equation}
The set $\mathcal S$, characterized by inequality \eqref{eq.Qk}, comprises all system matrices that are consistent with the collected offline input–state data and therefore contains the true system, i.e., $(A_s,B_s)\in\mathcal{S}$. 
Moreover,the resulting set $\mathcal S$ may be smaller than each individual consistency set $\Xi_k$. 
The unknown system matrices are characterized through this representation, which is constructed from noisy data following the same line of reasoning as in \cite{XieAuto,noisydata,HuCDC}. In contrast to \cite{XieAuto,noisydata,HuCDC}, however, the representation \eqref{eq.Qk} exploits not only the collected input–state data but also the computed value $Z(\bar{x}_k)$.

The objective is to stabilize the origin of the unknown system \eqref{eq.system} subject to prescribed state constraint. To formalize the control objective, we consider the following quadratic stage cost
\begin{equation}
	l(x, u)= \| x \|_E^2+\| u \|_R^2,
\end{equation}
where preselected weighting matrices $E \succ 0$ and $R \succ 0$. While we focus on stabilization at the origin for notational simplicity, the proposed framework can be extended to non-zero equilibria. Ellipsoidal constraint on predicted states is imposed as 
\begin{equation}\label{eq.xsxx}
	\left\| x_k \right\|_{S_x}^2 \le 1
\end{equation}
with  $S_x\succ0$. Furthermore, building on the data-based system representation, we formulate a data-driven min–max MPC problem with infinite horizon:
\begin{subequations}\label{cost_uio1}
	\begin{align}
		\underset{\tilde u_k}{\mathop {\min}}
		\underset{(A, B)\in \mathcal{S}}{\mathop {\max}}&~~
		\sum\limits_{t=0}^{\infty}l(\tilde x_k(t), \tilde u_k(t) ) \label{con1_uio1} \\
		{\rm{s.t.}} ~~&~  ~\tilde x_{k}(t+1)=AZ(\tilde x_k(t))+B \tilde u_k(t), \label{con2_uio1} \\ 
		& ~~  \tilde x_k(0)= x_k, \label{con3_uio1} \\
		&~  \left\| \tilde x_k(t) \right\|_{S_x}^2 \le 1 \ \ \forall (A, B) \in \mathcal{S}, \ t \in \mathbb{Z}. \label{con4_uio1} 
	\end{align}
\end{subequations}
Here, $\tilde{x}_k(t)$ and $\tilde{u}_k(t)$ denote the predicted state and control input at time $k + t$, respectively. The control objective is to minimize the worst-case accumulated cost over all system matrices $(A, B) \in \mathcal{S}$, by selecting the input sequence $\tilde u_k(t)$ with $t=0$ to $t=\infty$. The dynamics use the nominal form with $(A,B)\in\mathcal S$, and the vector $Z(\tilde x_k(t))$ can be computed from the predicted state $\tilde x_k(t)$ and the known nonlinear function $Q(\tilde x_k(t))$. The initial state is set to the measured state $x_k$ at time $k$ in constraint \eqref{con3_uio1}. Moreover, the predicted state is required to satisfy the ellipsoidal state constraint in \eqref{con4_uio1} for all $t \in \mathbb{Z}$, uniformly over all system matrices $(A, B) \in \mathcal{S}$.

Directly solving this min–max problem is intractable. Inspired by classical min–max MPC \cite{Bemporadmm2003,MPCLMI1996,Mayne1998}, we restrict our attention to a time-varying state-feedback control law of the form 
\begin{equation}
	\tilde u_k(t)=K_kZ(\tilde x_k(t))
\end{equation}
with a time-varying gain $K_k\in\mathbb{R}^{n_u\times n_z}$ to be designed. This structure retains the nonlinear feature $Z(\cdot)$ in the feedback while allowing a tractable formulation. Note that prior data-driven min–max MPC methods for linear systems \cite{XieAuto,YifanX} typically use state-feedback $\tilde u_k(t)=K_k\tilde x_k(t)$, but our approach accounts for the nonlinear lifting. We now formalize the control synthesis problem.
\begin{problem}\label{prob.1}
	Consider the nonlinear system \eqref{eq.system} under Assumptions \ref{Ass.system}–\ref{Ass.noise}. Given the offline input-state measurements \eqref{eq.XU}, design a data-driven min–max MPC scheme that synthesizes a state-feedback control law ensuring that the closed-loop system satisfies the state constraint.
\end{problem}

\section{Data-driven Min-max MPC}\label{sec.measure}
As discussed in Section \ref{sec.prosta}, the offline data $\bar U$, $\bar X$ and $\bar Z$ subject to disturbances satisfying Assumption \ref{Ass.noise}, is employed to construct a data-driven set $\mathcal{S}$ that characterizes all system matrices consistent with the behavior of the system \eqref{eq.system}. Leveraging this parametrization, we formulate a data-driven min-max MPC problem with state constraints \eqref{cost_uio1}. 

To solve problem \ref{prob.1}, we consider a simplified setting wherein the state data measured online during closed-loop operation is assumed to be noise-free. The objective is to design a data-driven SDP for the unknown system 
$$x_{k+1} = A_s Z(x_k) + B_sK_kZ(x_k)$$
to stabilize the origin. At each time step, a state-feedback control law is computed through a receding horizon approach by solving the corresponding SDP. We further establish that the proposed SDP guarantees recursive feasibility and ensures exponential stability of the closed-loop system. 

\subsection{Data-driven Min-max MPC in the Noise-free Case}\label{sec.3.2}
Our objective is to derive an upper bound on the worst-case cost over all system matrices in the set $\mathcal{S}$, and subsequently synthesize a state-feedback control law that minimizes this bound. To this end, we consider a quadratic Lyapunov function $V(\tilde x_k(t)) = \left\| \tilde x_k(t) \right\|_{P_k}^2$ with $P_k \succ 0$ for any $ t  \in {\mathbb Z}$. 
Suppose that $V(\tilde x_k(t))$ satisfies the following inequality
\begin{align}
	V(\tilde x_k(t+1)) \!- \! V(\tilde x_k(t)) \! \le \! -\left\| \tilde u_k(t) \right\|_{R}^2 \! - \! \left\| \tilde x_k(t) \right\|_{E}^2 \label{eq.xt+1-t}
\end{align}
for all $\tilde{x}_k(t)$,$\tilde u_k(t)$ and $Z(\tilde{x}_k(t))$ predicted by \eqref{con2_uio1} with any $(A, B) \in \mathcal{S}$ that additionally satisfy the nonlinearity constraints in Assumption \ref{Ass.non}.
Following the same arguments in \cite{YifanX}, if \eqref{eq.xt+1-t} holds, then $V(x_k)$ serves as an upper bound on the  cost in \eqref{cost_uio1}, i.e.,
\begin{equation}\label{eq.v1}
	\underset{(A, B)\in {\mathcal{S}}}{\mathop {\max}}
	\sum\limits_{t=0}^{\infty}l(\tilde x_t(k), K(k)\tilde x_t(k) ) \! \le \! V(\tilde x_k(0)) \! = \!  V(x_k).	
\end{equation}
To ensure the existence of such an upper bound, the measured state $x_k$ must satisfy
\begin{equation}
	\left\| x_k \right\|_{P_k}^2 \le \gamma_k, \label{eq.xpxr}
\end{equation}
where $\gamma_k$ is a time-varying constant that upper bounds the optimal cost of \eqref{cost_uio1} at time $k$.
\begin{remark}
In inequality \eqref{eq.xpxr}, we introduce an upper bound $\gamma_k$ on the worst-case cost. This bound provides a theoretical guarantee that the objective function does not exceed a specified limit. While this bound is generally conservative, it remains analytically tractable and preserves robustness. Nevertheless, because the upper bound is not always tight, some conservatism is inevitably introduced in the control design.
\end{remark}
As discussed above \eqref{eq.xt+1-t}--\eqref{eq.xpxr}, we leverage the Lyapunov decrease condition \eqref{eq.xt+1-t} to derive an upper bound on the worst-case cost of the min–max control problem \eqref{cost_uio1}. In this setting, problem \eqref{cost_uio1} reduces to finding a state-feedback gain $K_k$ and a positive definite matrix $P_k \succ 0$ such that inequality \eqref{eq.zABK} holds for all $\tilde{x}_k(t)$ and $Q(\tilde{x}_k(t))$ satisfying the nonlinearity constraint in Assumption \ref{Ass.non}.
This nonlinearity constraint can be expressed as
\begin{align}\label{eq.tildeWSG}
	\begin{bmatrix}
		\tilde x_k(t) \\
		Q(\tilde x_k(t))
	\end{bmatrix}^{\top}
	\begin{bmatrix}
		- W & S \\
		S^\top & G
	\end{bmatrix}
	\begin{bmatrix}
		\tilde x_k(t) \\
		Q(\tilde x_k(t))
	\end{bmatrix}
	\le 0 .
\end{align}
Under this formulation, the following result holds. Solving the resulting data-driven SDP \eqref{sdp} in Theorem~\ref{The.3} yields a stabilizing state-feedback control law that minimizes an upper bound on the worst-case cost of the min–max control problem \eqref{cost_uio1}.

\begin{theorem}\label{The.3} $\mathbf{(SDP \ for  \ data}$-$\mathbf{driven \ min}$–$\mathbf{max \ MPC  }$ 
	$\mathbf{with \ noise}$-$\mathbf{free \ online  \ measurements)}$
	Consider the nonlinear system \eqref{eq.system} subject to Assumptions~\ref{Ass.system}–\ref{Ass.noise}. 
	Consider the decision variables
	$\gamma_k > 0$, $\lambda \in \mathbb{R}^{N}_+$, $P_k\succ0$, $K_k$, and nonnegative scalar 
	$\tau_k \ge 0$. Define auxiliary matrices $O_{k} = \begin{bmatrix} H_k & 0 \\ 0 & \gamma_kI \end{bmatrix} =\begin{bmatrix} \gamma_k P_k^{-1} & 0 \\ 0 & \gamma_k I \end{bmatrix} \in \mathbb R^{{n_z} \times {n_z}} $, $L_{k} = K_kO_{k}\in \mathbb R^{{n_u} \times {n_z}}$, $Y_{k} = \tau_k H_k$ and $\mu_{k} =\tau_k \gamma_k$.
	If there exists a feasible solution to the following SDP \eqref{sdp}, where $M_E^\top M_E =E$, $M_R^\top M_R =R$, $M_W^\top M_W =W$, and $\Pi(\lambda) =\sum\limits_{k=0}^{N-1}\lambda_k \mathcal {\bar Q}_k =\sum\limits_{k=0}^{N-1}\lambda_k
	\left[ \begin{matrix}
		I    &   \bar x_{k+1}  \\
		0  &  -Z(\bar x_k)  \\
		0  &   - \bar u_k  \\
	\end{matrix} \right] \left[ \begin{matrix}
		\varpi^2I &   0  \\
		0 &    -I  \\
	\end{matrix} \right] \left[ * \right]^\top$ \\
then i) the Lyapunov decrease condition \eqref{eq.xt+1-t} is satisfied, and, ii)
	the state-feedback control law $u_k=K_kZ(x_k)$ guarantees the optimal cost of \eqref{cost_uio1} to be upper bounded by $\gamma_k$.
	\begin{figure*}[!h]
		\normalsize
		\begin{subequations}\label{sdp}
			\begin{align}
				& \underset{\gamma_k, \mu_{k},\lambda \atop L_{k}, H_{k},  Y_{k}  }{\mathop{\text{minimize}}} \  \gamma_k \label{eq.con_1}\\
				& \quad \quad {\rm {s.t.}}   \left[ \begin{matrix}
					1  &  x_k^\top \\	x_k &  H_k  \\
				\end{matrix} \right] \succeq 0,  \quad \quad \left[ \begin{matrix}
					H_k  & H_k \\	H_k  & S_x^{-1}  \\
				\end{matrix} \right] \succeq 0, \label{eq.con_2}\\
				&\quad \quad \quad ~ \left[ \begin{matrix}
					\left[ \begin{matrix}
						-H_k   &  0  \\	0 & 0  \\ \end{matrix} \right]  +    \Pi(\lambda) &    \left[ \begin{matrix}
						0	\\ O_{k} \\ L_{k}
					\end{matrix} \right]   &  0  & 0 \\ \\
					\left[ \begin{matrix}
						0	 & O_{k}  &  L_{k}^\top\\
					\end{matrix} \right]   &   	\left[ \begin{matrix}
						- H_k  & -  Y_{k}S\\
						-  S^\top Y_{k}  & -\mu_{k} G \end{matrix} \right]  &  
					\left[ \begin{matrix}
						L_{k}^\top M_R^\top 	 &   \left[ \begin{matrix}
							H_kM_E^\top	 \\  0 
						\end{matrix} \right] \\
					\end{matrix} \right]
					& \left[ \begin{matrix}
						Y_{k}M_W^\top \\ 0
					\end{matrix} \right] \\ \\
					0   &  \left[ \begin{matrix}
						M_RL_{k} \\ \left[ \begin{matrix}
							M_EH_{k}  & 0 \\
						\end{matrix} \right]
					\end{matrix} \right]   &  -\gamma_k I & 0  \\ \\
					0  &  \left[ \begin{matrix}
						M_WY_{k}	 &  0 \\
					\end{matrix} \right] & 0 & -\mu_{k} I
				\end{matrix} \right] \preceq  0,  \label{eq.LMIdelta} \\
				&\quad  \quad \quad ~~  \gamma_k > 0, \quad \mu_{k}\ge 0,\quad  \lambda =( \begin{matrix}
					\lambda_0, & \cdots , & \lambda_{N-1}  \\
				\end{matrix} )\ge 0. \label{eq.con_mu} 
			\end{align}	
		\end{subequations}
		\hrulefill
		\vspace*{4pt}
	\end{figure*}
\end{theorem}

\begin{pf}
	As established in \eqref{eq.v1}–\eqref{eq.xpxr}, the optimal cost for \eqref{cost_uio1} is guaranteed to be at most upper bounded by $\left\| x_k \right\|_{P_k}^2$, which in turn is constrained by the scalar bound $\gamma_k$.  By introducing the matrix variable $H_k$ and applying the Schur complement, the inequality in \eqref{eq.xpxr} can be equivalently reformulated as the first LMI in \eqref{eq.con_2}.  The state constraint \eqref{con4_uio1} can be reformulated as the second LMI in \eqref{eq.con_2}, following a derivation analogous to that in \cite{YifanX}. 
	
	Furthermore, we derive a set of LMIs that ensure the inequality \eqref{eq.xt+1-t} holds. Specifically, by substituting \eqref{con2_uio1} into \eqref{eq.xt+1-t}, the inequality can be equivalently expressed as
	\begin{align}\label{eq.zABK}
		&	\left[ \begin{matrix}
			\tilde x_k(t)  \\
			Q(\tilde x_k(t) )  \\
		\end{matrix} \right]^{\top}\big((A+BK_k)^{\top}P_k(A+BK_k)-\left[ \begin{matrix}
			P_k ~&~ 0  \\
			0 ~&~ 0  \\
		\end{matrix} \right] \nonumber \\
		+&\left[ \begin{matrix}
			E ~&~ 0  \\
			0 ~&~ 0  \\
		\end{matrix} \right]+K_k^{\top}RK_k \big)\left[ \begin{matrix}
			\tilde x_k(t)  \\
			Q(\tilde x_k(t))  \\
		\end{matrix} \right] \le 0.	
	\end{align}
Furthermore, we require \eqref{eq.zABK} to hold for all $\tilde{x}_k(t)$ and 
$Q(\tilde{x}_k(t))$ predicted by \eqref{con2_uio1} that satisfy the nonlinear 
constraint \eqref{eq.tildeWSG}.  
Invoking the results in \cite[the proof of Theorem~1]{nqc} and 
\cite[Section~2.1.2]{nonVAY}, we obtain a Lyapunov decrease condition to which 
the lossless S-procedure applies. In particular, by \cite[Theorem~2.19]{nonVAY}, a necessary and sufficient condition for \eqref{eq.zABK} to hold for all such vectors satisfying \eqref{eq.tildeWSG} is that there exists a scalar $\tau_k \ge 0$ such that
\begin{align}
	& (A + BK_k)^\top P_k (A + BK_k)
	- \begin{bmatrix}
		P_k ~&~ 0 \\ 0 ~&~ 0
	\end{bmatrix}
	+ \begin{bmatrix}
		E ~&~ 0 \\ 0 ~&~ 0
	\end{bmatrix} \nonumber \\
+	&~ K_k^\top R K_k
	+ \tau_k
	\begin{bmatrix}
		W & -S \\ -S^\top & -G
	\end{bmatrix}
	\preceq 0.
	\label{eq.A+BK}
\end{align}
It should be noted that in \eqref{eq.A+BK}, the application of the lossless S-procedure yields a necessary and sufficient condition, and hence does not incur any additional conservatism.
	By pre- and post-multiplying inequality \eqref{eq.A+BK} by $\begin{bmatrix} P_k^{-1} & 0 \\ 0 &  I \end{bmatrix}$ and $\gamma_k $, respectively, the inequality \eqref{eq.A+BK} can be equivalently expressed as
	\begin{align} \label{eq.APA+H}
		& \left( A O_{k} + B L_{k} \right)^\top P_k \left( A O_{k} + B L_{k} \right)+ L_{k}^\top R L_{k} \nonumber  \\
		+&	  \begin{bmatrix}
			H_k (E+\tau_kW) H_k \!-\! \gamma_k H_k ~& -\gamma_k \tau_k H_kS\\
			-\gamma_k \tau_k S^\top H_k  ~& -\gamma_k^2\tau_kG
		\end{bmatrix}  \preceq 0,
	\end{align}
	where $O_k$ and $L_{k}$ are defined in Theorem \ref{The.3}.
	Multiplying \eqref{eq.APA+H} with $\frac{1}{\gamma_k}$ yields
	\begin{align}\label{eq.APA+H1}
		& \left( A O_{k} \!+\! B L_{k} \right)^\top \! H_k^{-1} \! \left( A O_{k} \!+\! B L_{k} \right) \!+ \! \frac{1}{\gamma_k} L_{k}^\top R L_{k}  \nonumber  \\
		+&	  \begin{bmatrix}
			\frac{1}{\gamma_k} H_k (E+\tau_kW) H_k - H_k ~&~ - \tau_k H_kS\\
			- \tau_k S^\top H_k  ~&~ -\gamma_k\tau_kG
		\end{bmatrix} \preceq 0.
	\end{align}
	Applying the Schur complement twice to \eqref{eq.APA+H1}, it follows that
	\begin{align}
		- H_k-\left( AO_{k} +  BL_{k} \right)\xi_k^{-1} \left( A{{O}_{k}}+BL_{k} \right)^{\top }  \preceq 0 \label{eq.A_sH11}
	\end{align}
	and $\xi_k^{-1} =  \left\{ {\begin{bmatrix}
			\frac{1}{\gamma_k} H_k (E \! +\!  \tau_kW) H_k  \!- \! H_k  & - \tau_k H_kS\\
			- \tau_k S^\top H_k  & -\gamma_k\tau_kG
	\end{bmatrix}} \right.$
	$\left. {+  \frac{1}{\gamma_k} L_{k}^\top R L_{k}} \right\}^{-1}\preceq 0.$ The inequality \eqref{eq.A_sH11} is equivalent to
	\begin{align}
		\left[ \begin{matrix}
			I   \\
			A^\top  \\
			B^\top  \\
		\end{matrix} \right]^\top \!\! \left[ \begin{matrix}
			H_k  &  0  \\
			0  &  	\left[ \begin{matrix}
				O_{k}  \\
				L_{k}  \\
			\end{matrix} \right]\xi_k^{-1}\left[ \begin{matrix}
				O_{k}  \\
				L_{k}  \\
			\end{matrix} \right]^\top  \\
		\end{matrix} \right] \!\! \left[ \begin{matrix}
			I   \\
			A^\top  \\
			B^\top   \\
		\end{matrix} \right] \succeq 0. \label{eq.augIAB}
	\end{align}
We require that inequality \eqref{eq.augIAB} hold simultaneously for all $(A,B) \in \Xi_k$ in \eqref{eq.Qk} with $k = 0,1,\cdots,N-1$, i.e., for all matrices $(A,B)$ in the intersection of the consistency sets $\Xi_k$, namely $(A,B) \in \mathcal S$.  
By applying the lossy S-procedure over the data-based set $\mathcal S$ \cite[Lemma~2]{noisydata}, a sufficient condition can be obtained if there exists a vector of nonnegative scalars
$\lambda_i \geq 0, \forall i=0, \ldots, N-1$ such that the following inequality holds 
\begin{align}
&	\left[ \begin{matrix}
		H_k & 0 \\
		0 & \begin{bmatrix} O_k \\ L_k \end{bmatrix} \xi_k^{-1} \begin{bmatrix} O_k \\ L_k \end{bmatrix}^\top
	\end{matrix} \right] - \lambda _0 \mathcal {\bar Q}_0 -   \cdots -  \lambda _{N-1} \mathcal {\bar Q}_{N-1}  \nonumber \\
 =& \left[ \begin{matrix}
 	H_k & 0 \\
 	0 & \begin{bmatrix} O_k \\ L_k \end{bmatrix} \xi_k^{-1} \begin{bmatrix} O_k \\ L_k \end{bmatrix}^\top
 \end{matrix} \right] - \sum\limits_{k=0}^{N-1}\lambda_k \mathcal {\bar Q}_k \succeq 0. \label{eq.xi}
\end{align}
Defining
$	\Pi(\lambda) $ as in Theorem \ref{The.3} and rewriting inequality \eqref{eq.xi} leads to
	\begin{align}\label{eq.xi1}
		\left[ \begin{matrix}
			-H_k  ~&~  0  \\
			0  ~&~  	0 \\
		\end{matrix} \right]+ \Pi(\lambda)- 	\left[ \begin{matrix}
			0   \\
			O_{k}  \\
			L_{k}  \\
		\end{matrix} \right]\xi_k^{-1}\left[ \begin{matrix}
			0   \\
			O_{k}  \\
			L_{k}  \\
		\end{matrix} \right]^\top \preceq 0.
	\end{align}
	Applying the Schur complement to \eqref{eq.xi1}, we obtain that
	\begin{equation}\label{eq.LMIxi}
		\left[ \begin{matrix}
			\left[ \begin{matrix}
				-H_k   ~&~  0  \\	0 & 0  \\ \end{matrix} \right]  +    \Pi(\lambda) ~&~    \left[ \begin{matrix}
				0	\\ O_{k} \\ L_{k}
			\end{matrix} \right]     \\
			\left[ \begin{matrix}
				0	 &~ O_{k}  &  L_{k}^\top\\
			\end{matrix} \right]   ~&~   	\xi_k
		\end{matrix} \right] \preceq  0.
	\end{equation}
	To facilitate further derivation, we reorganize the expression for $\xi_k={\begin{bmatrix}
			\frac{1}{\gamma_k} H_k (E +  \tau_kW) H_k  \!-\! H_k  &~ - \tau_k H_kS\\
			- \tau_k S^\top H_k  &~ -\gamma_k\tau_kG
	\end{bmatrix}} +  \frac{1}{\gamma_k} L_{k}^\top R L_{k} $. By introducing the notation $Y_{k}$ and $\mu_{k}$ as defined in Theorem \ref{The.3}, $\xi_k$ can be equivalently expressed as
	\begin{align}\label{eq.xik}
		\xi_k & =   {\begin{bmatrix}
				\frac{1}{\gamma_k} H_kEH_k +	\frac{1}{\mu_k} Y_{k} W  Y_{k} -H_k  & - Y_{k}S\\
				-  S^\top Y_{k}  & -\mu_k G
		\end{bmatrix}} \ \nonumber \\
		& \quad + \frac{1}{\gamma_k} L_{k}^\top R L_{k} \nonumber \\
		& = {\begin{bmatrix}
				- H_k  & - Y_{k}S\\
				-  S^\top  Y_{k}  & -\mu_k G
		\end{bmatrix}}+  \frac{1}{\gamma_k} L_{k}^\top R L_{k} \nonumber \\
		& \quad +\frac{1}{\gamma_k}{\begin{bmatrix}
				H_kEH_k  &~  0 \\
				0 &~ 0
		\end{bmatrix}}+\frac{1}{\mu_{k}}{\begin{bmatrix}
				Y_{k}W Y_{k}  &~  0 \\
				0 &~ 0
		\end{bmatrix}}.
	\end{align}
	Substituting \eqref{eq.xik} into \eqref{eq.LMIxi} and applying the Schur complement once more yields the LMI \eqref{eq.LMIdelta}.
	Taken together, this analysis confirms that under the conditions stated in Theorem \ref{The.3}, the optimal cost of problem \eqref{cost_uio1} is guaranteed to be upper bounded by $\gamma_k$.
\end{pf}

To control the unknown nonlinear system \eqref{eq.system} without disturbances, the SDP \eqref{sdp} is solved in a receding-horizon framework, wherein the optimal state-feedback gain is updated at each time step based on current measurements, as outlined in Algorithm 1. The optimal solution of \eqref{sdp}, denoted by $(\gamma^*_k, L^*_{k}, H^*_{k}, \lambda^*, Y^*_{k},\mu^*_{k})$, depends explicitly on the measured state $x_k$ and the associated feature vector $Z(x_k)$. The resulting control input is then computed using the state-feedback law $u_k = K^*_kZ(x_k)$, where $K^*_k$ is computed by $L^*_{k}\left(\begin{bmatrix} H_k^* & 0 \\ 0 & \gamma_k^*I \end{bmatrix}\right)^{-1}$.

{\centering
	\begin{tabular}{ll}
		\toprule
		\multirow{1}{*}{Algorithm 1: Data-driven min-max MPC.} \\
		\midrule
		\multirow{1}{*}{1. At time $k$, measure state $x_k$ and calculate vector }  \\ 
		{\quad $Z(x_k)$ using \eqref{eq.z}  } \\
		{2. Solve the problem \eqref{sdp} and obtain $K^*_k$} \\
		{3. Apply the input $u_k=K^*_k Z(x_k)$} \\
		{4. Set $k=k+1$ and go back to $1$} \\
		\bottomrule
	\end{tabular}\\}
\begin{remark}\label{re.xq}
	The SDP constructed in Theorem \ref{The.3} accounts for the case where the vector $Z(x_k)$ comprises both linear and nonlinear components. In the special case where $Z(x_k) = x_k$, the system \eqref{eq.system} reduces to a standard linear time-invariant system. 
	In this case, Assumption \ref{Ass.non} is no longer required, and the resulting min-max MPC problem reduces to \cite{YifanX}. In the opposite case where $Z(x_k) = Q(x_k)$ consists solely of nonlinear terms, the Lyapunov decrease condition, unlike the general inequality \eqref{eq.zABK}, explicitly contains no cross-terms between $\tilde{x}_k(t)$ and $Q(\tilde{x}_k(t))$. As a result, it suffices to impose the simplified version of the nonlinearity constraint in Assumption \ref{Ass.non}, i.e., the special case with $S = 0$, on all pairs $\tilde{x}_k(t)$ and $Q(\tilde{x}_k(t))$. Leveraging the S-procedure, a sufficient condition can then be derived to ensure that \eqref{eq.xt+1-t} holds for all such vectors. The resulting SDP can be formulated by following the same procedure as in Theorem \ref{The.3}, with the derivation simplified due to the absence of cross-terms.
\end{remark}
\begin{remark}\label{re.input}
	The optimization problem formulated in this paper considers only state constraint, without incorporating input constraint. Unlike the settings in \cite{XieAuto,YifanX}, input constraint is not considered here, primarily because deriving an LMI that guarantees $\| \tilde{u}_k(t) \|_{S_u}^2 \leq 1$ with $S_u \succ 0$ is substantially more challenging.
	Specifically, the approach in \cite{XieAuto,YifanX}, which employs the S-procedure to jointly enforce the inequalities $\| x_k \|_{P_k}^2 \leq \gamma_k$ and $\| \tilde{u}_k(t) \|_{S_u}^2 \leq 1$, cannot be directly applied in our setting. This difficulty arises because the state-feedback law considered here takes the form $\tilde{u}_k(t) = K_k Z(\tilde{x}_k(t))$, rather than the standard state-feedback structure $\tilde{u}_k(t) = K_k \tilde{x}_k(t)$. Extending the proposed framework to explicitly incorporate input constraints therefore remains an important direction for future research.
\end{remark}
\begin{remark}
If the chosen function dictionary is misspecified, it may capture only partial or inaccurate information about the system dynamics. This mismatch leads to an inaccurate evaluation of the lifted historical data $\bar Z$, which in turn compromises the construction of the data-driven set-membership representation \eqref{eq.Qkk}. Consequently, when Theorem \ref{The.3} is invoked, the resulting SDP may yield a control gain $K_k^{*}$ that fails to provide adequate closed-loop stabilization, potentially leading to degraded performance.
\end{remark}
\subsection{Closed-loop guarantees}\label{sec.3.3}
The following theorem establishes the recursive feasibility of the SDP problem \eqref{sdp} and demonstrates that the closed-loop system $x_{k+1} = A_s Z(x_k) + B_s K_k Z(x_k)$ is exponentially stabilized to the origin. Moreover, it is shown that the closed-loop state trajectory satisfies the prescribed state constraint.
\begin{theorem}\label{Th.guarantee}
	Let Assumptions \ref{Ass.system}--\ref{Ass.noise} hold. If the SDP problem \eqref{sdp} is feasible at initial time $k=0$, then 
	
	\begin{itemize}
		\item[i)] it remains feasible for all $k\ge0$;  
		\item[ii)] the closed-loop trajectory satisfies the state constraint, i.e., $\| x_k \|_{S_x}^2 \le 1$ ; and, 
		\item[iii)]  the closed-loop system $x_{k+1}=A_sZ(x_k)+B_sK_kZ(x_k)$ converges exponentially to the origin.
	\end{itemize}
	
\end{theorem}

\begin{pf} 
	(i) Recursive feasibility
	
	Assume that the SDP \eqref{sdp} is feasible at time step $k$. Let $K_k^*$ and $P_k^*$ denote the optimal state-feedback gain and the corresponding Lyapunov matrix obtained from the solution to \eqref{sdp}. Substituting $K_k = K_k^*$ and $P_k = P_k^*$ into  \eqref{eq.A+BK}, it follows that for all $(A, B) \in \mathcal{S}$, the following inequality holds
	\begin{align}\label{eq.P*K*1}
		& (A+BK_k^*)^{\top }P_k^*(A+BK_k^*)-\left[ \begin{matrix}
			P_k^* ~&~ 0  \\
			0 ~&~ 0  \\
		\end{matrix} \right] \nonumber \\
		\preceq & -\left[ \begin{matrix}
			E ~&~ 0  \\
			0 ~&~ 0  \\
		\end{matrix} \right]-(K_k^*)^{\top }RK_k^* -\tau_k \begin{bmatrix}
			W & -S \\ -S^\top  & -G
		\end{bmatrix}.	
	\end{align}
	Since constraint \eqref{eq.WSG} holds for all $\tilde{x}_k(t)$ and $Q(\tilde{x}_k(t))$, it follows that the matrix 
	$\begin{bmatrix}
		W & -S \\ -S^\top  & -G
	\end{bmatrix} \succeq 0$. Moreover, given that $R \succ 0$ and $\tau_k \ge 0$, the inequality \eqref{eq.P*K*1} implies that
	\begin{align}\label{eq.P*K*}
		& (A+BK_k^*)^{\top }P_k^*(A+BK_k^*)-\left[ \begin{matrix}
			P_k^* ~&~ 0  \\
			0 ~&~ 0  \\
		\end{matrix} \right]  \\
		\preceq & -\left[ \begin{matrix}
			E ~&~ 0  \\
			0 ~&~ 0  \\
		\end{matrix} \right] \! - \! (K_k^*)^{\top }RK_k^* \! - \! \tau_k \begin{bmatrix}
			W & -S \\ -S^\top  & -G
		\end{bmatrix} \! \preceq  \! -\left[ \begin{matrix}
			E ~&~ 0  \\
			0 ~&~ 0  \\
		\end{matrix} \right].	\nonumber
	\end{align}
	Due to $(A_s, B_s) \in \mathcal{S}$, pre- and post-multiplying both sides of the inequality \eqref{eq.P*K*} by $Z^\top(x_k)$ and $Z(x_k)$, respectively, yields
	\begin{align*}
		& \big((A_s+B_sK_k^*)Z(x_k)\big)^{\top }P_k^*(A_s+B_sK_k^*)Z(x_k)\nonumber \\
		&	-  Z^\top(x_k) \left[ \begin{matrix}
			P_k^* ~&~ 0  \\
			0 ~&~ 0  \\
		\end{matrix} \right]Z(x_k)\prec -Z^\top(x_k)\left[ \begin{matrix}
			E ~&~ 0  \\
			0 ~&~ 0  \\
		\end{matrix} \right]Z(x_k),	
	\end{align*}
	which implies that
	\begin{align}\label{eq.under}
		\left\| x_{k+1} \right\|_{P_k^*}^2 - \left\| x_{k} \right\|_{P_k^*}^2 \le -\underline{\lambda}(E)\left\| x_{k} \right\|_{2}^2 \le 0,
	\end{align}
	where $\underline{\lambda}(E)$ denote the smallest eigenvalue of $E$.
	Leveraging \eqref{eq.xpxr}, it follows that 
	\begin{equation}\label{eq.xpk}
		\left\| x_{k+1} \right\|_{P_k^*}^2  \le \left\| x_{k} \right\|_{P_k^*}^2  \le \gamma_k^*. 
	\end{equation}
	Consequently, the optimal solution at time $k$ remains feasible for the problem \eqref{sdp} at time $k+1$, thereby establishing recursive feasibility.
	
	(ii) Constraint satisfaction
	
	The proof of constraint satisfaction proceeds in two stages. First, we establish that
	$$\Pi_{RPI} =  \left\{  x \in \mathbb R^{n_x}: 	\| x\|_{P_{k}}^2 \le \gamma_k \right\}$$
	is a RPI set for all $(A, B) \in \mathcal {S}$. Second, we show that if the second LMI in \eqref{eq.con_2} is satisfied, then every state $x \in \Pi_{RPI}$ automatically satisfies the prescribed state constraint. These two steps follow an argument analogous to that in Lemma 1 and Theorem 3 of \cite{YifanX}, respectively.

	(iii)  Exponential stability
	
	Before proving the exponential stability, we first derive upper and lower bounds on the Lyapunov function $V(x_k) = \|x_k\|_{P_k^*}^2$.
	From the inequality \eqref{eq.under}, we obtain a lower bound as
	\begin{equation}\label{eq.low}
		\underline{\lambda}(E) \|x_k\|_2^2  \le \|x_{k}\|_{P_k^*}^2  - \|x_{k+1}\|_{P_k^*}^2  \le \|x_k\|_{P_k^*}^2.
	\end{equation}
	Furthermore, from (i), note that $P_{k}^*$ is a feasible solution to the SDP \eqref{sdp} at time step $k+1$, whereas $P_{k+1}^*$ denotes the optimal solution, it follows that
	\begin{align}\label{eq.xx}
		\|x_{k+1}\|_{P_{k+1}^*}^2 \le \|x_{k+1}\|_{P_{k}^*}^2,
	\end{align}
	which holds for $k \in \mathbb{Z}$. This implies that $\|x_k\|_{P_k^*}^2 \le \|x_k\|_{P_0^*}^2$.
	Together, these bounds yield the inequality
	$$\underline{\lambda}(E)\|x_k\|_2^2 \le \|x_k\|_{P_k^*}^2 \le \|x_k\|_{P_0^*}^2.$$
	Combining \eqref{eq.xx} with \eqref{eq.under} results in
	$$\|x_{k+1}\|_{P_{k+1}^*}^2 \! - \!  \|x_k\|_{P_k^*}^2 \le \|x_{k+1}\|_{P_k^*}^2  \!-\! \|x_k\|_{P_k^*}^2 \! \le \!  -\underline{\lambda}(E)\|x_k\|_2^2,$$
	which can be rewritten as
	$$V(x_{k+1}) - V(x_k) \le -\underline{\lambda}(E)\|x_k\|_2^2.$$
	This proves that the origin is exponentially stable under the proposed closed-loop control scheme. 
\end{pf}

\section{Data-driven Robust Min-max MPC \\ using noisy data}\label{sec.noisy}
The data-driven min-max MPC proposed in Section~\ref{sec.measure} relies on the assumption of noise-free online state measurements. However, this assumption is often violated in practice due to sensor noise and modeling uncertainties, which limits the applicability of the proposed framework in real-world systems. This motivates the development of a robust controller that explicitly accounts for noisy online state data, described by
$$x_{k+1} = A_s Z(x_k) + B_sK_kZ(x_k)+d_k.$$ 
When disturbances are present in both the offline data-collection phase and the online execution phase, the SDP constructed in Theorem~\ref{The.3}, although exhibiting a certain degree of robustness, may lead to degraded closed-loop performance. In particular, for sufficiently large disturbances, the closed-loop trajectories may lose stability guarantees. This can be attributed to the fact that constraint~\eqref{eq.A+BK} is derived under the assumption of noise-free online state measurements and therefore does not explicitly capture the influence of online disturbances on the closed-loop dynamics. As a result, the exponential stability guarantee established in Section~3 cannot be extended to ensure robust stability in the presence of online disturbances.

To overcome this limitation, we explicitly incorporate the impact of online disturbances on the closed-loop dynamics. Specifically, a new inequality, given in~\eqref{eq.barAsno}, is derived to enforce the decrement condition~\eqref{eq.A+BK} while explicitly accounting for bounded disturbances. Building on this condition, we formulate a new data-driven SDP and provide robust stability guarantees in the sense of the existence of a robust positively invariant (RPI) set around the origin, as well as recursive feasibility of the optimization problem. The overall control architecture is illustrated in Fig.~\ref{fig.111}.
\begin{figure}[h]
\centering
	\includegraphics[width=0.47\textwidth]{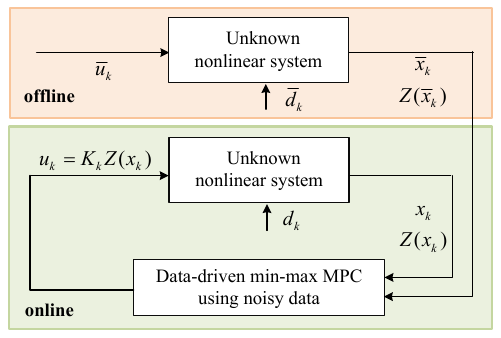}\caption{Control diagram of the proposed data-driven min-max MPC scheme.}
	\label{fig.111}
\end{figure}

\subsection{Data-driven Min-max MPC in the Noisy Case}\label{sec.4.1}
To address Problem \ref{prob.1}, we reformulate the problem \eqref{cost_uio1} into a computationally tractable SDP that accommodates both noisy offline input-state data and noisy online state measurements. The resulting SDP yields an upper bound on the optimal cost over the data set $\mathcal{S}$ consistent with the offline data and concurrently synthesizes a state-feedback gain that minimizes this bound.

\begin{theorem}\label{The.4} 
$\mathbf{(SDP \ for  \ data}$-$\mathbf{driven \ min}$–$\mathbf{max \ MPC  }$ 
$\mathbf{with \ noisy \ online  \ measurements )}$ Consider the nonlinear system \eqref{eq.system} subject to Assumptions~\ref{Ass.system}–\ref{Ass.noise}. 
	Let the decision variables be 
	$\gamma_k > 0$, $\lambda \in \mathbb{R}^{N}_+$, $P_k\succ0$, $K_k$,  along with nonnegative scalars 
	$\{\tau_k \ge 0, \alpha >  {\underline{\lambda}({ E})}\}$. Define auxiliary matrices $O_{k} = \begin{bmatrix} H_k & 0 \\ 0 & \gamma_kI \end{bmatrix} =\begin{bmatrix} \gamma_k P_k^{-1} & 0 \\ 0 & \gamma_k I \end{bmatrix} \in \mathbb R^{{n_z} \times {n_z}} $, $L_{k} = K_kO_{k}\in \mathbb R^{{n_u} \times {n_z}}$, $Y_{k} = \tau_k H_k$,  and $\mu_{k} =\tau_k \gamma_k$.
	If there exist a feasible solution for the SDP \eqref{sdpnoise}, where $M_E^\top M_E =E$, $M_R^\top M_R =R$, $M_W^\top M_W =W$,  and $\Pi(\lambda) =\sum\limits_{k=0}^{N-1}\lambda_k \mathcal {\bar Q}_k=\sum\limits_{k=0}^{N-1}\lambda_k
	\left[ \begin{matrix}
		I    &  \bar x_{k+1}  \\
		0   &  -Z(\bar x_k)  \\
		0  &   - \bar u_k  \\
	\end{matrix} \right] \left[ \begin{matrix}
		\varpi^2I  &    0  \\
		0  &    -I  \\
	\end{matrix} \right] \left[ \begin{matrix}
		I &  \bar x_{k+1}  \\
		0 & - Z(\bar x_k)  \\
		0 & -u_k  \\
	\end{matrix} \right]^\top$ \\
  then the following hold: i) the Lyapunov decrease condition \eqref{eq.xt+1-t} is satisfied; and ii) the state-feedback control law $u_k = K_k Z(x_k)$ guarantees that the optimal cost of \eqref{cost_uio1} is upper bounded by $\gamma_k$.
	\begin{figure*}[!h]
		\normalsize
		\begin{subequations}\label{sdpnoise}
			\begin{align}
				& \underset{\gamma_k, \mu_{k},\lambda \atop L_{k}, H_{k},  Y_{k}}{\mathop{\text{minimize}}} \  \gamma_k \label{eq.con_1_no}\\
				& \quad \quad {\rm {s.t.}}   \left[ \begin{matrix}
					1  &  x_k^\top \\	x_k &  H_k  \\
				\end{matrix} \right] \succeq 0,  \quad \quad \left[ \begin{matrix}
					H_k  & H_k \\	H_k  & S_x^{-1}  \\
				\end{matrix} \right] \succeq 0, \label{eq.con_2_no}\\
				&\quad \quad \quad ~ \left[ \begin{matrix}
					\left[ \begin{matrix}
						-H_k + \frac{\gamma_k}{\alpha}I  &  0  \\	0 & 0  \\ \end{matrix} \right]  +    \Pi(\lambda) &    \left[ \begin{matrix}
						0	\\ O_{k} \\ L_{k}
					\end{matrix} \right]   &  0  & 0 \\ \\
					\left[ \begin{matrix}
						0	~&~ O_{k}  ~&~  L_{k}^\top\\
					\end{matrix} \right]   &   	\left[ \begin{matrix}
						- H_k  \!&\! -  Y_{k}S\\
						-  S^\top Y_{k}  \!&\! -\mu_{k} G \end{matrix} \right]  &  
					\left[ \begin{matrix}
						L_{k}^\top M_R^\top 	 &   \left[ \begin{matrix}
							H_kM_E^\top	 \\  0 
						\end{matrix} \right] \\
					\end{matrix} \right]
					& \left[ \begin{matrix}
						Y_{k}M_W^\top \\ 0
					\end{matrix} \right] \\ \\
					0   &  \left[ \begin{matrix}
						M_RL_{k} \\ \left[ \begin{matrix}
							M_EH_{k}  & 0 \\
						\end{matrix} \right]
					\end{matrix} \right]   &  -\gamma_k I & 0  \\ \\
					0  &  \left[ \begin{matrix}
						M_WY_{k}	 &  0 \\
					\end{matrix} \right] & 0 & -\mu_{k} I
				\end{matrix} \right] \preceq  0,  \label{eq.LMIdelta_no} \\
				&\quad  \quad \quad ~~  \gamma_k > 0, \quad \mu_{k}\ge 0,\quad  \lambda =( \begin{matrix}
					\lambda_0, & \cdots , & \lambda_{N-1}  \\
				\end{matrix} )\ge 0. \label{eq.con_mu_no} 
			\end{align}	
		\end{subequations}
		\hrulefill
		\vspace*{4pt}
	\end{figure*}
\end{theorem}

\begin{pf}
	Applying the Schur complement, the inequality constraint in \eqref{eq.xpxr} can be equivalently rewritten as the first LMI in \eqref{eq.con_2_no}. Moreover, in analogy with Theorem \ref{The.3}, it is necessary to ensure that the inequality \eqref{eq.A+BK} is satisfied for all system matrices $( A, B)\in {\mathcal{S}}$. However, due to online noisy state, ensuring that inequality \eqref{eq.A+BK} holds for nominal system is insufficient to guarantee the robust stability of the closed-loop system. To address this limitation, we assume that all predicted vectors $Z(\tilde x_k(t))$ satisfy the following inequality
	\begin{align}\label{eq.real}
		&	\big(\left(A+  B K_k\right) Z(\tilde x_k(t)) + d_k \big)^\top P_k [*] - \alpha d_k^\top d_k \nonumber \\
		  - & Z^\top(\tilde x_k(t)) \left[ \begin{matrix}
			P_k ~&~ 0  \\
			0 ~&~ 0  \\
		\end{matrix} \right]  Z(\tilde x_k(t)) \nonumber \\
	 \le &	- Z^\top(\tilde x_k(t)) \left(\left[ \begin{matrix}
			E ~&~ 0  \\
			0 ~&~ 0  \\
		\end{matrix} \right]+K_k^{\top}RK_k\right)Z(\tilde x_k(t)),
	\end{align} 
	which is equivalent to
	\begin{equation}\label{eq.xsigma1}
		\left[ \begin{matrix}
			Z(\tilde x_k(t))  \\
			d(k)  \\
		\end{matrix} \right]^\top  \left[ \begin{matrix}
			\Omega_1 &  (A+BK_k)^\top P_k \\
			P_k(A+BK_k) & P_k-\alpha I  \\
		\end{matrix} \right] \left[ * \right]   \le  0,
	\end{equation}
	where $\Omega_1 =(A+BK_k)^{\top}P_k(A+BK_k)-\left[ \begin{matrix}
		P_k ~&~ 0  \\
		0 ~&~ 0  \\
	\end{matrix} \right] +\left[ \begin{matrix}
		E ~&~ 0  \\
		0 ~&~ 0  \\
	\end{matrix} \right]+K_k^{\top}RK_k $. 
Furthermore, we require \eqref{eq.xsigma1} to hold for all $Z(\tilde{x}_k(t))$  that satisfy the nonlinear 
	constraint \eqref{eq.tildeWSG}.  By invoking the lossless S-procedure \cite{nqc}, inequality \eqref{eq.xsigma1} holds if and only if there exists a scalar $\tau_k \ge 0$ such that
\begin{align}\label{eq.barAsno}
	\begin{bmatrix}
		\Omega_1 & (A + BK_k)^\top P_k \\
		P_k(A + BK_k) & P_k - \alpha I
	\end{bmatrix} 
	\!- \! 
	\begin{bmatrix}
		\tau_k \begin{bmatrix} -W & S \\ S^\top & G \end{bmatrix} &~ 0 \\
		0 &~ 0
	\end{bmatrix} \! \preceq \! 0,
\end{align}
which can be equivalently rewritten as
\begin{align}\label{eq.barAs}
	\begin{bmatrix}
		\Omega_1 - \tau_k \begin{bmatrix} -W & S \\ S^\top & G \end{bmatrix} ~&~ (A + BK_k)^\top P_k \\
		P_k(A + BK_k) ~&~ P_k - \alpha I
	\end{bmatrix} \preceq 0.
\end{align}
	We aim to demonstrate two key points: i) inequality \eqref{eq.barAs} can be reformulated as an LMI, and ii) the satisfaction of \eqref{eq.barAs} implies that the decrement condition \eqref{eq.A+BK} is also satisfied. Together, these results ensure that the condition \eqref{eq.A+BK} holds while explicitly accounting for the impact of process disturbances.
	
	By applying the Schur complement to \eqref{eq.barAs}, we obtain a set of equivalent matrix inequalities
	\begin{align}
		&  P_k -\alpha I \preceq  0, \label{eq.p<a} \\
		& 	\Omega_1  - \tau_k  
		\left[ \begin{matrix}
			-W & S  \\
			S^\top  & G  \\
		\end{matrix} \right] \nonumber \\
		- & (A+BK_k)^\top P_k(P_k - \alpha I)^{-1}P_k(A+BK_k) \preceq  0. \label{eq.sc1} 
	\end{align}
	The inequality \eqref{eq.sc1} can then be expressed as a LMI using the historical dataset $\mathcal{S}$. The derivation follows similar steps to those in Theorem \ref{The.3}; hence, we provide only a sketch of the proof. Following \eqref{eq.APA+H}--\eqref{eq.A_sH11}, inequality \eqref{eq.sc1} can be arranged as
	\begin{align}\label{eq.APA+H11}
		& \left(H_k^{-1}-\frac{1}{\gamma_k}P_k(P_k -\alpha I)^{-1}P_k \right)^{-1} \!- \! \left( A O_{k} \! +\! B L_{k} \right)      \nonumber  \\
		\times	&\left\{ \! -  \frac{1}{\gamma_k} L_{k}^\top R L_{k} \! +	\!  \begin{bmatrix}
			H_k	\!\!-\!\! \frac{1}{\gamma_k} H_k (E \!\! +\! \! \tau_kW) H_k   \!&~ \tau_k H_kS\\
			\tau_k S^\top H_k  \!&~ \gamma_k\tau_kG
		\end{bmatrix} \right\}^{-1} \nonumber \\
		\times	&\left( A O_{k} + B L_{k} \right)^\top  \succeq  0.
	\end{align}
	Using the Woodbury matrix identity \cite{Wood1950}, the first term in \eqref{eq.APA+H11} simplifies to
	$$\left(H_k^{-1}-\frac{1}{\gamma_k}P_k(P_k -\alpha I)^{-1}P_k \right)^{-1} = H_k -\frac{\gamma_k}{\alpha}I.$$
	Substituting this expression into \eqref{eq.APA+H11} and using the notation $\xi_k$ and $F_k$ defined in \eqref{eq.A_sH11} and Theorem \ref{The.4} result in
	\begin{align}\label{eq.xi_no}
		& H_k \! - \! \frac{\gamma_k}{\alpha}I \!- \! \left( A O_{k} + B L_{k} \right) (-\xi_k)^{-1} \left( A O_{k} + B L_{k} \right)^\top  \succeq  0,
	\end{align}
	which is equivalent to 
	\begin{align}
		\left[ \begin{matrix}
			I   \\
			A^\top  \\
			B^\top  \\
		\end{matrix} \right]^\top \!\! \left[ \begin{matrix}
			H_k -\frac{\gamma_k}{\alpha}I &  0  \\
			0  &  	\left[ \begin{matrix}
				O_{k}  \\
				L_{k}  \\
			\end{matrix} \right]\xi_k^{-1}\left[ \begin{matrix}
				O_{k}  \\
				L_{k}  \\
			\end{matrix} \right]^\top  \\
		\end{matrix} \right] \!\! \left[ \begin{matrix}
			I   \\
			A^\top  \\
			B^\top   \\
		\end{matrix} \right] \succeq 0. \label{eq.augIAB_no}
	\end{align}
	Following similar derivation steps as in \eqref{eq.xi}–\eqref{eq.xik}, the LMI in \eqref{eq.LMIdelta_no} is obtained.
	
	We next show that if inequality \eqref{eq.barAs} holds, then the decrease condition \eqref{eq.A+BK} is also satisfied. Substituting \eqref{eq.p<a} into \eqref{eq.sc1} and rearranging terms yields
	\begin{align*} 
		& \Omega_1  - \tau_k  
		\left[ \begin{matrix}
			-W & S  \\
			S^\top  & G  \\
		\end{matrix} \right] \nonumber \\
		\preceq & (A+BK_k)^\top P_k(P_k -\alpha I)^{-1}P_k(A+BK_k) \preceq 0. 
	\end{align*}
	Given the matrix 
	$\begin{bmatrix}
		-W & S \\ S^\top  & G
	\end{bmatrix} \preceq 0$ and scalar  $\tau_k \geq 0$, it follows from
	\begin{align*} 
		& \Omega_1   \preceq \tau_k  
		\left[ \begin{matrix}
			-W & S  \\
			S^\top  & G  \\
		\end{matrix} \right] \preceq 0
	\end{align*}
	that the following holds 
	\begin{align*} 
		& (A+BK_k)^{\top}P_k[*] \! - \!  \left[ \begin{matrix}
			P_k ~&~ 0  \\
			0 ~&~ 0  \\
		\end{matrix} \right] \! + \! \left[ \begin{matrix}
			E ~&~ 0  \\
			0 ~&~ 0  \\
		\end{matrix} \right]+K_k^{\top}RK_k \preceq 0.
	\end{align*}
	This construction ensures that the closed-loop system satisfies the decrement condition in \eqref{eq.A+BK}. 
	
	Additionally, the state  constraint in \eqref{con4_uio1} can be reformulated as data-based LMIs given in \eqref{eq.con_2_no}. The reformulation follows the same argument as in Theorem \ref{The.3} and are omitted here for brevity. Consequently, if Theorem \ref{The.4} holds, the optimal cost in \eqref{cost_uio1} is guaranteed to be at most $\|x_k\|_{P_k}^2$, and this cost is upper bounded by $\gamma_k$.
\end{pf}

\subsection{Closed-loop guarantees}\label{sec.4.2}
The following theorem establishes the recursive feasibility of the optimization problem \eqref{sdpnoise}, as well as robust stability of the resulting closed-loop system $x_{k+1} = A_s Z(x_k) + B_sK_kZ(x_k)+d_k$  from the data-driven SDP scheme. Furthermore, we also prove that the state constraint is satisfied for the closed-loop trajectory.

\begin{theorem}\label{Th.2} 
	Suppose that Assumptions \ref{Ass.system}---\ref{Ass.noise} are satisfied and the optimization problem \eqref{sdpnoise} is feasible at time $k \! = \! 0$ with $\gamma^*_{0,m} \! \ge \! \frac{\alpha^2\varpi^2}{\underline{\lambda}({E})}$. Let 
	$\Pi_{ROA} =  \left\{  x \in \mathbb R^{n_x}: 	\| x\|_{P^*_{0,m}}^2 \le \gamma^*_{0,m} \right\}$ and
	$\Pi_{RPI} =  \big\{  x \in \mathbb R^{n_x}: $ $ 	\| x\|_{P^*_{k,m}}^2 \le \frac{\alpha^2\varpi^2}{\underline{\lambda}({E})} \big\}.$ Then,
	\begin{itemize} 
		\item[i)] the optimization problem \eqref{sdpnoise} is feasible for any 	states $x_k \in \Pi_{ROA}  \backslash  \Pi_{RPI}$; 
		\item[ii)] the closed-loop trajectory satisfies the state constraint, i.e., $\| x_k \|_{S_x}^2 \le 1$ ; and, 
		\item[iii)] the closed-loop system is robustly stabilized to a robust positive invariant set, i.e., $\Pi_{RPI}$.
	\end{itemize}
\end{theorem}

\begin{pf}The proof is similar to that of Theorem 2 in \cite{XieAuto}, so we only provide a sketch of the proof here.   
	
	(i) Recursive feasibility
	
	Suppose that problem \eqref{sdpnoise} is feasible at time $k$. We first establish both lower and upper bounds on the function $V(x_k)$. Since Theorem \ref{The.3} ensures that inequality \eqref{eq.v1} holds, combining it with \eqref{con3_uio1} yields
	\begin{align}\label{eq.low}
		V(x_k) & \ge \underset{(A, B)\in {\mathcal{S}}}{\mathop {\max}}
		\sum\limits_{t=0}^{\infty}l(\tilde x_t(k), K(k)\tilde x_t(k) ) \\ \nonumber
		& \ge  l(x_k, u_k) \ge \underline{\lambda}(E) \| x_k\|_2^2.
	\end{align}
	Applying \eqref{eq.p<a} provides the following upper bound on $V(x_k)$
	\begin{equation}\label{eq.upbound}
		\| x_k\|_{P_{k}}^2   \le \alpha\| x_k\|_2^2.
	\end{equation}
	Then, combining these results, we obtain that
	\begin{equation}\label{eq.bound}
		\underline{\lambda}(E) \| x_k\|_2^2 \le \| x_k\|_{P_{k}}^2  \le \alpha\| x_k\|_2^2.
	\end{equation}
	Moreover, as established in Theorem \ref{The.4}, the matrix inequality \eqref{eq.barAs} holds. Denote by $K_{k,m}^*$ and $P_{k,m}^*$ the optimal state-feedback gain and the associated Lyapunov matrix, respectively, obtained from the solution of \eqref{sdpnoise}.
	Due to system matrices $(A_s,B_s) \in \mathcal{S}$, substituting $K_k = K_{k,m}^*$ and $P_k = P_{k,m}^*$ into  \eqref{eq.barAs}, 
	and pre- and post-multiplying the resulting expression by $\left[ \begin{matrix}
		Z(x_k)^\top & d_k^\top \\
	\end{matrix} \right]$ and its transpose, respectively, yields
	\begin{align}\label{eq.Vhatdelta1}
		&	\big(\left(A_s+  B_s K^*_{k,m}\right) Z(x_k) + d_k \big)^\top P^*_{k,m} [*] \nonumber \\
		+ & Z^\top(x_k) \left(\left[ \begin{matrix}
			E &~ 0  \\
			0 &~ 0  \\
		\end{matrix} \right] \!+ \! (K^*_{k,m})^{\top}RK^*_{k,m} \!- \! \tau_k  
		\left[ \begin{matrix}
			-W &  S  \\
			S^\top & G  \\
		\end{matrix} \right] \right)[*]  \nonumber \\
		- & \alpha d_k^\top d_k - Z^\top( x_k) \left[ \begin{matrix}
			P^*_{k,m} &~ 0  \\
			0 &~ 0  \\
		\end{matrix} \right] [*]  \le  0	.
	\end{align}

	Given that $R \succ 0$, $\tau_k \ge 0$ and $\begin{bmatrix}
		-W & S \\ S^\top  & G
	\end{bmatrix} \preceq 0$, the inequality \eqref{eq.Vhatdelta1} implies that
	\begin{align}\label{eq.Vhatdelta}
		&	\big(\left(A_s+  B_s K^*_{k,m}\right) Z(x_k) + d_k \big)^\top P^*_{k,m} [*] - x^\top_k P^*_{k,m}x_k \nonumber \\
		\le &  Z^\top(x_k) \left(-\left[ \begin{matrix}
			E ~&~ 0  \\
			0 ~&~ 0  \\
		\end{matrix} \right] \!- \! (K^*_{k,m})^{\top}RK^*_{k,m} \!+ \! \tau_k  
		\left[ \begin{matrix}
			-W &  S  \\
			S^\top & G  \\
		\end{matrix} \right] \right)  \nonumber \\
		& \times Z(x_k) + \alpha d_k^\top d_k \nonumber \\
		\le & -Z^\top(x_k) \left[ \begin{matrix}
			E ~&~ 0  \\
			0 ~&~ 0  \\
		\end{matrix} \right]Z(x_k) + \alpha d_k^\top d_k \nonumber \\
		= & -x^\top_k E x_k + \alpha d_k^\top d_k.
	\end{align}

	Recalling Assumption \ref{Ass.noise} and inequality \eqref{eq.bound}, the inequality \eqref{eq.Vhatdelta} above can be expressed as
	\begin{align}\label{eq.x}
		&	\| x_{k+1}\|_{P^*_{k,m}}^2 - 	\| x_k\|_{P^*_{k,m}}^2  \\
		\le & 	-  \underline{\lambda}(E)\| x_k\|_2^2 +  \alpha\varpi^2 \le  - \frac{\underline{\lambda}(E)}{\alpha} \| x_k\|_{P^*_{k,m}}^2 + \alpha\varpi^2.  \nonumber
	\end{align} 
	Subtracting $ \frac{\alpha^2\varpi^2}{\underline{\lambda}(E)}$ from both sides of the inequality \eqref{eq.x} and adding $\| x_k\|_{P^*_{k,m}}^2$ to both sides, we obtain that
	\begin{align}\label{eq.Ly}
		\| x_{k+1}\|_{P^*_{k,m}}^2  \! - \!  \frac{\alpha^2\varpi^2}{\underline{\lambda}(E)}
		 \! \le \!  \left(1  \! - \!  \frac{\underline{\lambda}({E})}{\alpha}\right) \! \left(\| x_k\|_{P^*_{k,m}}^2  \! - \!  \frac{\alpha^2\varpi^2}{\underline{\lambda}(E)} \right). 
	\end{align}
	
	We now proceed to establish recursive feasibility. Consider the case where the state lies outside the set $\Pi_{RPI}$, i.e., $\frac{\alpha^2\varpi^2}{\underline{\lambda}({E})} < \| x_k\|_{P^*_{k,m}}^2 \le \gamma^*_{k,m} $. Noting that $\alpha > {\underline{\lambda}({E})}$, it follows that  $0<1  -  \frac{\underline{\lambda}({E})}{\alpha}<1$. Consequently, inequality \eqref{eq.Ly} leads to
	\begin{equation}\label{eq.xkk}
		\| x_{k+1}\|_{P^*_{k,m}}^2  -  \frac{\alpha^2\varpi^2}{\underline{\lambda}({E})} \le   \| x_k\|_{P^*_{k,m}}^2 -  \frac{\alpha^2\varpi^2}{\underline{\lambda}({E})}.
	\end{equation}
Then, we have that
	\begin{equation}\label{eq.xkkk}
	\| x_{k+1}\|_{P^*_{k,m}}^2  \le   \| x_k\|_{P^*_{k,m}}^2 \le \gamma^*_{k,m},
\end{equation}
	which implies that the feasible solution of the optimization problem \eqref{sdpnoise} at time $k$ is also feasible at time $k+1$.
	
	(ii) Constraint satisfaction
	
	In this part, we separately consider the cases $x_k \in \Pi_{ROA}  \backslash  \Pi_{RPI}$ and $x_k \in  \Pi_{RPI}$ . In both cases, the proofs of state constraint satisfaction for states outside and within the RPI set follow the same reasoning as in \cite[Theorem 2]{XieAuto}.
	
	(iii) Robust stability
	
	Building upon the preceding analysis, we establish robust stability of the closed-loop system. When state $ x_k$ is outside the set $\Pi_{RPI}$, i.e., $x_k \in \Pi_{ROA}  \backslash  \Pi_{RPI}$, the fact that $P_k^*$ is a feasible solution and $P_{k+1}^*$ is the optimal solution to the SDP \eqref{sdpnoise} at time step $k+1$ implies, via inequalities \eqref{eq.xkk}--\eqref{eq.xkkk}, that
	\begin{align*}
		& \| x_{k+1}\|_{P^*_{{k+1},m}}^2  - \frac{\alpha^2\varpi^2}{\underline{\lambda}({E})}	\nonumber \\
		 \le & \| x_{k+1}\|_{P^*_{k,m}}^2  -  \frac{\alpha^2\varpi^2}{\underline{\lambda}({E})}  \le  \| x_k\|_{P^*_{k,m}}^2 -  \frac{\alpha^2\varpi^2}{\underline{\lambda}({E})}.
	\end{align*}
	When the state $x_k$ lies within the set  $ \Pi_{RPI}$, we further obtain that
	\begin{equation}
		\|x_{k+1}\|_{P^*_{{k+1},m}}^2   \! \le \!  \|x_{k+1}\|_{P^*_{k,m}}^2  \! \le  \! \|x_{k}\|_{P^*_{k,m}}^2  \! \le \! \frac{\alpha^2\varpi^2}{\underline{\lambda}({E})} \le \gamma^*_{0,m}.
	\end{equation}
	Thus, we conclude that the set $ \Pi_{RPI}$  is robustly stabilized for the closed-loop system.
\end{pf}
To regulate the unknown nonlinear system \eqref{eq.system} in the presence of online and offline disturbances, the SDP \eqref{sdpnoise} is solved within a receding-horizon control framework, as outlined in Algorithm 2.  At each time step $k$, optimal solution of the problem \eqref{sdpnoise} is denoted by $(\gamma^*_{k,m}, L^*_{k,m}, H^*_{k,m}, \lambda^*_m, Y^*_{k,m},\mu^*_{k,m})$, providing the optimal state-feedback gain 
\begin{equation*}
	K^*_{k,m}=L^*_{k,m}\left(\begin{bmatrix} H_{k,m}^* & 0 \\ 0 & \gamma_{k,m}^*I \end{bmatrix}\right)^{-1}
\end{equation*} and corresponding matrix $P^*_{k,m}$.

{\centering
	\begin{tabular}{ll}
		\toprule
		\multirow{1}{*}{Algorithm 2: Data-driven robust min-max MPC.} \\
		\midrule
		\multirow{1}{*}
		{~1. At time $k$, measure state $x_k$ and calculate $Z(x_k)$  }  \\ 
		{~2. Solve the problem \eqref{sdpnoise} and obtain $K^*_{k,m}$} \\
		{~3. Apply the input $u_k=K^*_{k,m}Z(x_k)$} \\
		{~4. Setting $k=k+1$, measure the state $x_k$ } \\
		{\quad ~ and calculate $Z(x_k)$} \\
		{~5. Solve the optimization problem \eqref{sdpnoise} } \\
		{~6. if   $\gamma^*_{k,m}> \frac{\alpha^2\varpi^2}{\underline{\lambda}({E})}$ then} \\
		{~7. ~~~Apply the input $u_k=K^*_{k,m}Z(x_k)$} \\
		{~8. ~~~Set $k \!= \! k \!+ \! 1$ and measure  $x_k$, go back to 5} \\
		{~9. else if   $\gamma^*_{k,m} \le \frac{\alpha^2\varpi^2}{\underline{\lambda}({E})}$ then} \\
		{10. \quad Set $K_{k}=K^*_{{k-1},m}$ and $P_{k}=P^*_{{k-1},m}$} \\
		{11. \quad Apply the input $u_{k}=K_{k}Z(x_k)$} \\
		{12. \quad Setting $k \!= \! k \! + \! 1$, measure $x_k$ and calculate $Z(x_k)$} \\
		{\quad \quad  ~ go back to 11} \\
		{13. end if} \\
		\bottomrule
	\end{tabular}}

\begin{remark}
	In this section, we reformulate the data-driven min–max MPC problem with noisy online measurements as an SDP, and establish recursive feasibility, constraint satisfaction, and robust stability of the resulting closed-loop system. While the derivation of the SDP formulation and the associated guarantees is inspired by \cite{XieAuto}, the proposed framework is tailored to nonlinear systems rather than linear ones, thereby introducing additional technical challenges. Specifically, the proposed SDP is designed to ensure not only that the Lyapunov decrease condition \eqref{eq.xt+1-t} is satisfied, but also that the nonlinear constraint \eqref{eq.WSG} holds. Furthermore, since ${x}_{k+1} = (A + B K_k) Z({x}_k)$ comprises linear and nonlinear components, the analysis must explicitly account for the effect of the nonlinear term. Although Theorem \ref{Th.2} establishes robust stability for the closed-loop system, the result holds only for systems satisfying Assumption \ref{Ass.non}, which imposes certain limitations. Future work will extend the proposed framework to more general classes of nonlinear systems.
\end{remark}
\begin{remark}
The proposed approaches in Theorems \ref{The.3} and \ref{The.4} exhibit some conservatism. Specifically, for both the noise-free and noisy online measurement scenarios, solving the SDP problems \eqref{sdp} or \eqref{sdpnoise} yields an upper bound on the objective of problem \eqref{cost_uio1} while guaranteeing satisfaction of the imposed nonlinear constraints. A primary source of conservatism arises from the restriction to a state-feedback control law of the form $u_k=B_s K_k Z(x_k)$. Moreover, the upper bound may not always be tight, and constraining the cost upper bound to be quadratic may further enlarge the gap. An additional source of conservatism arises from the use of the lossy S-procedure in \eqref{eq.augIAB} and \eqref{eq.augIAB_no} within Theorems~\ref{The.3} and~\ref{The.4}.
Specifically, in Theorem~1 (see \eqref{eq.augIAB}), the lossy S-procedure is employed to derive sufficient conditions under which inequality~\eqref{eq.augIAB} holds simultaneously for all $(A,B) \in \Xi_k$ defined in \eqref{eq.Qk}, for $k = 0,1,\ldots,N-1$. This relaxation is inherently conservative. The same reasoning applies to Theorem~\ref{The.3}.
In contrast, in \eqref{eq.A+BK} and \eqref{eq.barAs}, the lossless S-procedure is applied, yielding necessary and sufficient conditions and thus introducing no additional conservatism. The reduction of conservatism through more general state-feedback laws and less restrictive cost upper-bounding functions remains an open challenge and constitutes an important direction for future research.
\end{remark}
\begin{remark}
As established in Theorems \ref{The.4} and \ref{Th.guarantee}, the size of the RPI set $\Omega_{RPI}$ is governed by the preselected design parameters $\alpha$ and $E$, together with the disturbance bound $\varpi$. The parameter $\alpha$ must satisfy $\alpha > \underline{\lambda}(E)$. Smaller values of $\alpha$ and $\varpi$ yield a smaller ellipsoid invariant set to which the closed-loop system converges, while also corresponding to a smaller feasible region. As discussed in \cite{XieAuto}, choosing $\alpha$ overly small renders the feasible region excessively restrictive, which can degrade closed-loop performance or even lead to infeasibility at the initial time. Conversely, if $\alpha$ is chosen too large, the initial state may already lie inside the RPI set, causing the controller to effectively reduce to a static state-feedback law and resulting in poor performance. Moreover, in the proposed SDP formulation, the matrix $S_x$ in the state constraint \eqref{eq.con_2} or \eqref{eq.con_2_no} is user-specified. In particular, smaller values of $S_x$ correspond to a larger admissible state constraint set and do not introduce feasibility issues. Conversely, larger values of $S_x$ lead to a tighter state constraint set, which may render the SDP infeasible at the initial time step.
\end{remark}
\begin{remark}
In this work, we first consider the case of noise-free online measurements. For this setting, the constructed constraint \eqref{eq.A+BK} ensures that the optimal cost of \eqref{cost_uio1} is upper bounded by $\gamma_k$ while simultaneously satisfying the imposed nonlinear constraints. Building on this result, we then extend the proposed framework to explicitly account for noisy online measurements. This two-stage treatment provides a more comprehensive analysis and facilitates future research.
In particular, Theorem \ref{The.3} suggests natural extensions to other classes of nonlinear systems and alternative constraint structures, while Theorem \ref{The.4} motivates further investigations into robustness issues, including measurement errors and output-feedback control in scenarios where online measurements are unavailable.	
\end{remark}

\section{Simulation Results}\label{sec.sim}
We illustrate the proposed methods on two nonlinear systems. We compare our schemes (D-MPC for noise-free and D-RMPC for noisy) with two benchmarks: i) an identification-based robust min-max MPC (ID-RMPC) that first identifies the system matrices $(A,B)$ via least-squares based on the offline data and subsequently applies a robust min-max MPC; and, ii) a data-driven static feedback controller (D-FC) based on a contraction/Lyapunov condition \cite{HuZJ}. The simulations are performed on a desktop computer with an Intel Core i7-11700K processor (8 cores, 16 threads, 3.60 GHz) using MATLAB 2022a. All semidefinite programs are solved via YALMIP \cite{Yalmip}.

\subsection{Flexible-Joint Robot with Noise-free Online Measurements}\label{sec.sim1}
We consider an Euler-discretized model of a flexible-joint robot, adapted from \cite{robotmo}, given by
$$\left\{
\begin{aligned}
	\theta_{m,k+1} &= \theta_{m,k} + T_s \omega_{m,k} + d_{1,k}, \\
	\omega_{m,k+1} &= \omega_{m,k} + \frac{K_aT_s}{J_m}(\theta_{l,k}-\theta_{m,k})-\frac{B_mT_s}{J_m} \omega_{m,k}  \\
	& \quad +\frac{K_{\tau}T_s}{J_m}u_k +d_{2,k}, \\
	\theta_{l,k+1} &= \theta_{l,k} + T_s \omega_{l,k} +d_{3,k},\\
	\omega_{l,k+1} &= \omega_{l,k}- \frac{K_aT_s}{J_l}\theta_{l,k} + \frac{K_aT_s}{J_l} \theta_{m,k}, \\
	& \quad -\frac{mghT_s}{J_l}\sin(\theta_{l,k})
	+ d_{4,k}.
\end{aligned}
\right. $$
Here, $\theta_{m,k}$ and $\omega_{m,k}$ denote the angular position and angular velocity of the motor, respectively, while $\theta_{l,k}$ and $\omega_{l,k}$ denote the angular rotation and angular velocity of the link. The control input $u_k$ corresponds to the motor-applied torque. 
The sampling period is $T_s = 0.1\,\mathrm{s}$. The physical parameters are chosen as follows: the motor inertia is $J_m = 0.0037 \,\mathrm{kg m^2}$, the link inertia is $J_l = 0.0093\,\mathrm{kgm^2}$, the torsional spring constant is $K_a = 0.18 \,\mathrm{N m/rad}$, the viscous friction coefficient $B_m =0.046 \,\mathrm{N m/V}$, and the amplifier gain $K_{\tau} = 0.08\,\mathrm{N m/V}$. The nominal gravitational load acting on the link is $mgh$, with mass $m = 0.21\,\mathrm{kg}$, gravitational acceleration $g = 9.81\,\mathrm{m/s^2}$, and link length $h = 0.015\,\mathrm{m}$.

We define the state vector as 
$x_k=\big[ \theta_{m,k} \ \ \omega_{m,k} \ \ \theta_{l,k} \ \ \omega_{l,k} \big]^\top=\big[ x_{1,k} \ \ x_{2,k} \ \ x_{3,k} \ \ x_{4,k} \big]^\top$, 
the disturbance vector as 
$d_k=\big[ d_{1,k} \ \ d_{2,k} \ \ d_{3,k} \ \ d_{4,k} \big]^\top$, 
and introduce the lifting as
$Z(x_k)=\big[ x_k^\top \ \ \sin(x_{3,k}) \big]^\top$.
Following \eqref{eq.system1}--\eqref{eq.system}, the flexible-joint robot dynamics can be equivalently expressed in the state-space form as
\begin{align}\label{eq.robot}
	x_{k+1} =\; & 
	\left[ \begin{matrix}
		1 & T_s & 0 & 0 & 0  \\
		-\frac{K_a T_s}{J_m} & 1-\frac{B_m T_s}{J_m} & \frac{K_a T_s}{J_m} & 0 & 0  \\
		0 & 0 & 1 & T_s & 0  \\
		\frac{K_a T_s}{J_l} & 0 & -\frac{K_a T_s}{J_l} & 1 & -\frac{m g h T_s}{J_l}
	\end{matrix} \right] Z(x_k) \nonumber \\
	& + \left[ \begin{matrix}
		0 \ \  \frac{K_{\tau} T_s}{J_m} \ \ 0 \ \ 0
	\end{matrix} \right]^\top u_k + d_k .
\end{align}
For the proposed D-MPC and static feedback control D-FC \cite{HuZJ}, the offline data collection phase is subject to bounded process disturbances satisfying $\|d_k\|_2^2 \le 4 \times 10^{-4}$. We collect data along trajectories of length $N=25$, with randomly generated initial states and inputs whose components lie in the interval $[-1,1]$. The stage cost weights are chosen as $E = 0.1 I$ and $R = 0.1$. For the nonlinear constraint in \eqref{eq.WSG}, we set $S = 0.01 I$, $G = 0.01 I$, and $W = 0.01 I$. State constraints are imposed via \eqref{eq.con_2} and \eqref{eq.con_2_no}, with $S_x = 0.01 I$.  For the D-FC method, the parameter $\beta$ is set to $0.65$.

Fig.~\ref{fig.xno} compares the closed-loop state trajectories obtained under the proposed D-MPC and the D-FC schemes over 100 time steps for system without online process disturbances. The D-FC approach is applicable in both scenarios where only the offline data are corrupted by noise and where both offline and online data are noisy. As expected, all state trajectories converge to the origin. In view of \eqref{eq.robot}, the system state $x_k$ implicitly encodes the information of the lifted variable  $Z(x_k)$; consequently, the trajectories of $Z(x_k)$ also converge to the origin. Meanwhile, Fig.~\ref{fig.xs} illustrates the closed-loop state trajectory under the proposed D-MPC scheme, from which it can be observed that the state constraint \eqref{eq.xsxx} is satisfied.

To verify that the nonlinearities arising in the considered application satisfy the quadratic constraint in \eqref{eq.WSG}, Fig.~\ref{fig.non1} evaluates the nonlinear constraint along both the offline data trajectories and the online closed-loop trajectories generated by the two control schemes. We define 
$J_{nonlinear} = \begin{aligned} Z(x) ^\top \begin{bmatrix} -W & S \\ S^\top & G \end{bmatrix} Z(x) \end{aligned}$ in Fig. \ref{fig.non1}. It can be observed that, since the offline data are generated along open-loop trajectories that may be divergent, the value of $J_{nonlinear}$ increases over time. In contrast, during online operation, both control methods drive $Z(x_k)$ asymptotically to the origin, and hence $J_{nonlinear}$ converges to 0.

Finally, Figs.~\ref{fig.xno}--\ref{fig.non1} show that the proposed online controller achieves superior closed-loop performance compared with D-FC, which employs a static feedback gain obtained by solving the SDP only once using offline data. This performance improvement is attributed to the adaptive nature of the proposed D-MPC, which continuously updates the feedback gain using real-time state measurements, thereby enabling faster stabilization of the closed-loop system.
\begin{figure}[h]
	\centering
	\includegraphics[width=0.48\textwidth]{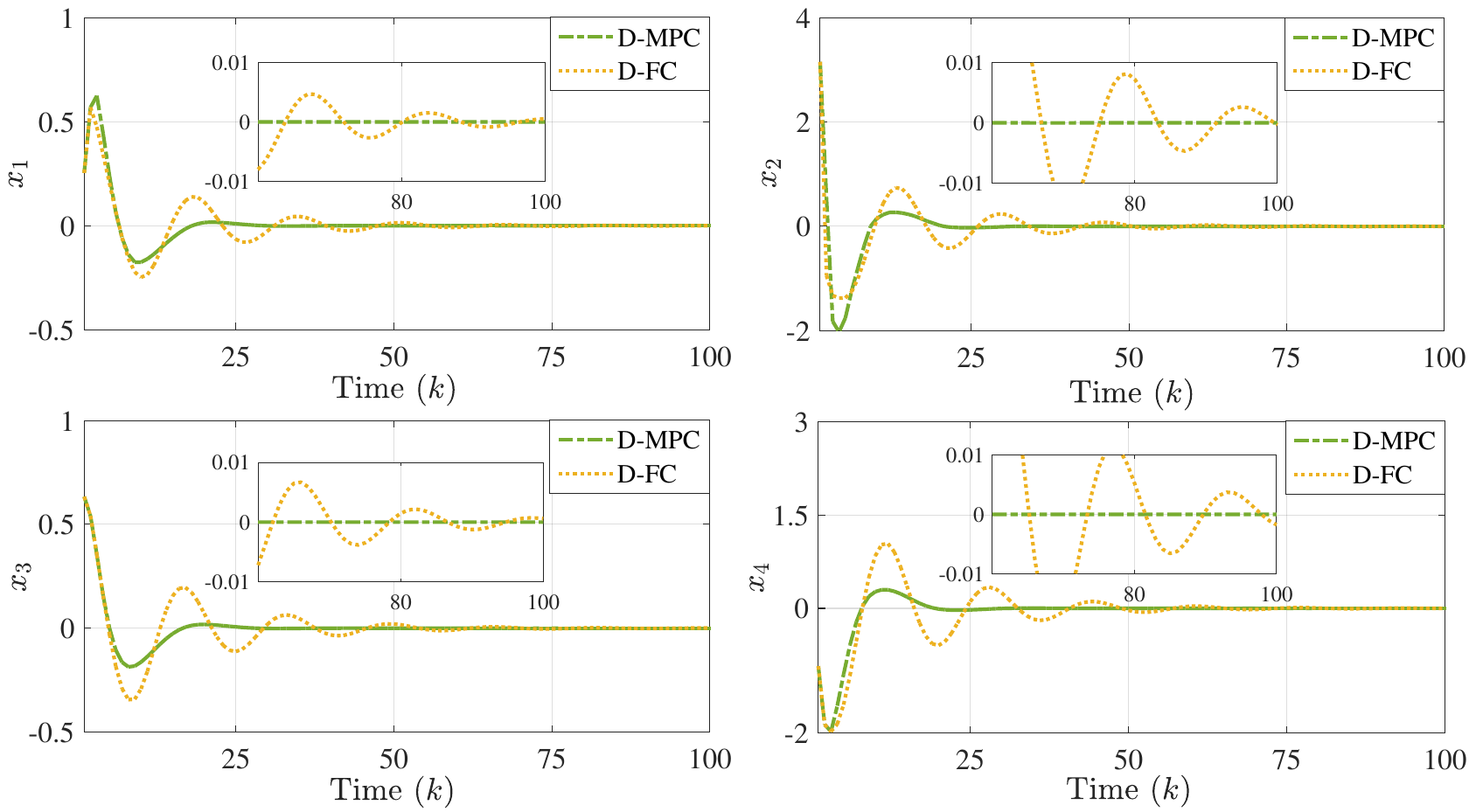} \caption{Closed-loop state trajectories under D-MPC and D-FC with noise-free online measurements}
	\label{fig.xno}
\end{figure}
\begin{figure}[h]
	\centering
	\includegraphics[width=0.4\textwidth]{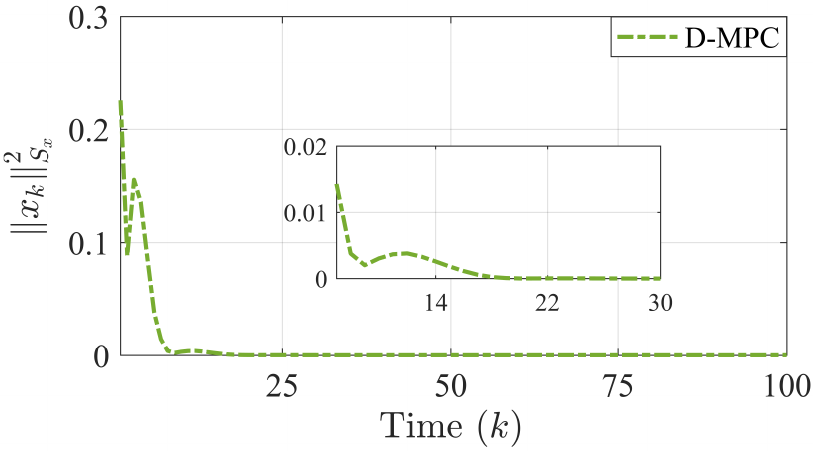} \caption{Closed-loop state constraint under D-MPC.}
	\label{fig.xs}
\end{figure}
\begin{figure}[h]
	\centering
	\includegraphics[width=0.48\textwidth]{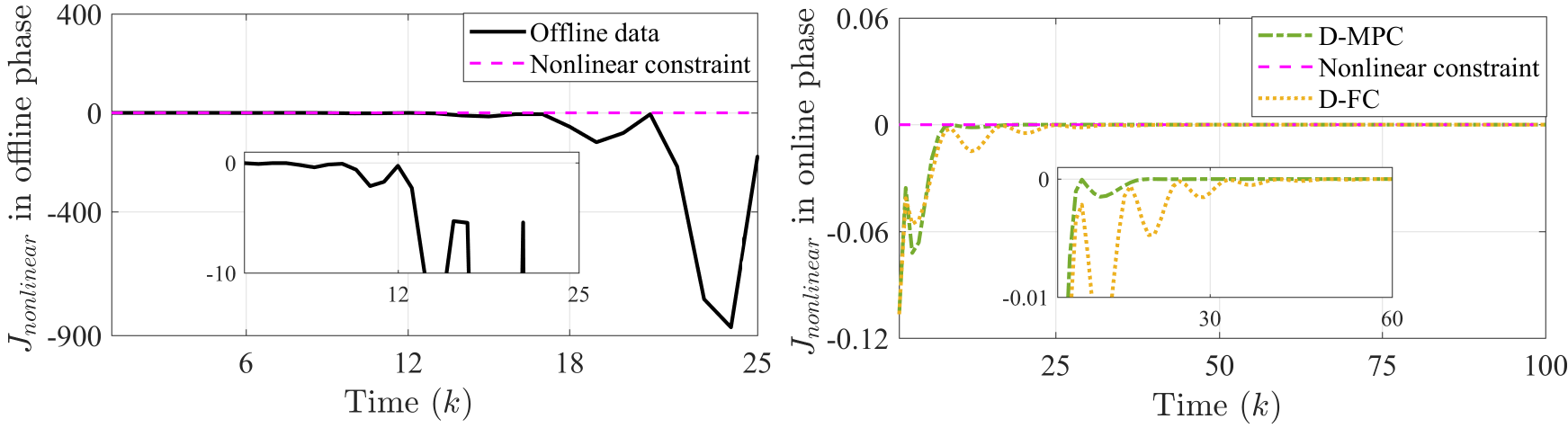} \caption{Nonlinear constraint during the offline collection phase and the online implementation phase}
	\label{fig.non1}
\end{figure}
To further assess the performance of the two controllers, we evaluate three metrics: the closed-loop stage cost, the computation time and nonlinear constraint cost, averaged over 50 independent Monte Carlo experiments. For a single experiment, the cumulative closed-loop stage cost over the horizon $[1, T]$ is defined as
$J = \sum_{k=1}^{T} l(x_k, u_k) = \sum_{k=1}^{T} \big( \|x_k\|_{E}^{2} + \|u_k\|_{R}^{2} \big)$.
The closed-loop cost corresponds to the empirical average over 50 experiments, i.e.,
$\bar{J}= \frac{1}{50} \sum_{i=1}^{50} J_i$,
where $J_i$ denotes the cost obtained in the $i$-th trial. The computation time is defined analogously as the average over 50 runs. 

Table \ref{tab.one} summarizes the closed-loop stage cost, and average computational time for the D-MPC and D-FC schemes. D-MPC attains the lowest closed-loop cost, reflecting its ability to stabilize the closed-loop system more rapidly. Specifically, D-MPC achieves a $36.91\%$ reduction in closed-loop cost relative to D-FC. These results demonstrate that the online controller D-MPC exhibits superior closed-loop control performance compared with the static controller D-FC.
{\begin{table}[ht]
		\caption{Evaluation metrics of four methods}
		\centering	
		\begin{tabular}{cccc} 
			\toprule
			\multirow{1}{*}{Scheme}&{closed-loop cost}&{}&{computation time (s)} \\
			\midrule 
			\multirow{1}{*}{D-MPC} & {3.2989}&{}&{13.7321}  \\ 
			\multirow{1}{*}{D-FC} &5.2286&{}&{/}  \\
			\bottomrule
			\label{tab.one}	
		\end{tabular}
\end{table}}

\subsection{Flexible-Joint Robot with Noisy Online Measurements}\label{sec.sim2}
In this section, we further evaluate the performance of the proposed D-RMPC and D-MPC schemes and compare them with the ID-RMPC and D-FC methods on a flexible-joint robot subject to disturbances during both the offline and online phases. Both the offline data collection phase and the online implementation are subject to bounded process disturbances satisfying $\|d_k\|_2^2 \le 4 \times 10^{-4}$.
For all four control methods, offline data are collected along trajectories of length $N = 25$, with randomly generated initial states and inputs whose components lie in the interval $[-1,1]$. For D-RMPC, D-MPC, and D-FC, the parameters are chosen consistently with those used in Section~\ref{sec.sim1}. In addition, the parameter in the constraint \eqref{eq.LMIdelta_no} is selected as $\alpha = 7 \times 10^{4}$. For the ID-RMPC method, the system matrices are identified via a least-squares procedure using the same offline dataset. During the online phase, the initial state is randomly selected from the set $[0,1]^4$.

Figs.~\ref{fig.x}–\ref{fig.non} present the closed-loop state trajectories and the evolution of the nonlinear constraint under different controllers, respectively. 
The signal-to-noise ratio (SNR) is evaluated by comparing the offline data matrices \eqref{eq.XU} and \eqref{eq.D}.
In Fig.~\ref{fig.x}, the SNR is $60.2$~dB, under which all considered methods achieve robust stabilization of the system.
Although the proposed D-MPC is designed under the assumption of noise-free state measurements during online operation, it remains effective in the presence of online disturbances, albeit with slower convergence compared to the robust schemes. In contrast, the proposed D-RMPC, which explicitly accounts for both offline and online noise, achieves the fastest convergence to a neighborhood around the origin. Moreover, by leveraging time-varying state-feedback policies, D-RMPC converges significantly faster than the static D-FC. The closed-loop system under ID-RMPC converges more slowly than D-RMPC, primarily due to model inaccuracies introduced by noisy offline data. Furthermore, as shown in Fig. \ref{fig.xs1}, the state constraint is satisfied throughout the closed-loop operation of the proposed D-RMPC, D-MPC, and ID-RMPC schemes.  Among the three methods, D-RMPC exhibits the fastest convergence of $\|x_k\|_{S_x}^2$ toward a neighborhood of the origin.
In Fig.~\ref{fig.non}, the offline data satisfy the nonlinear constraint. During the online phase, D-RMPC drives $Z(x_k)$ to the origin earlier than the other controllers, and consequently, the nonlinear constraint cost $J_{nonlinear}$ for D-RMPC asymptotically approaches zero faster.
\begin{figure}[h]
	\centering
	\includegraphics[width=0.48\textwidth]{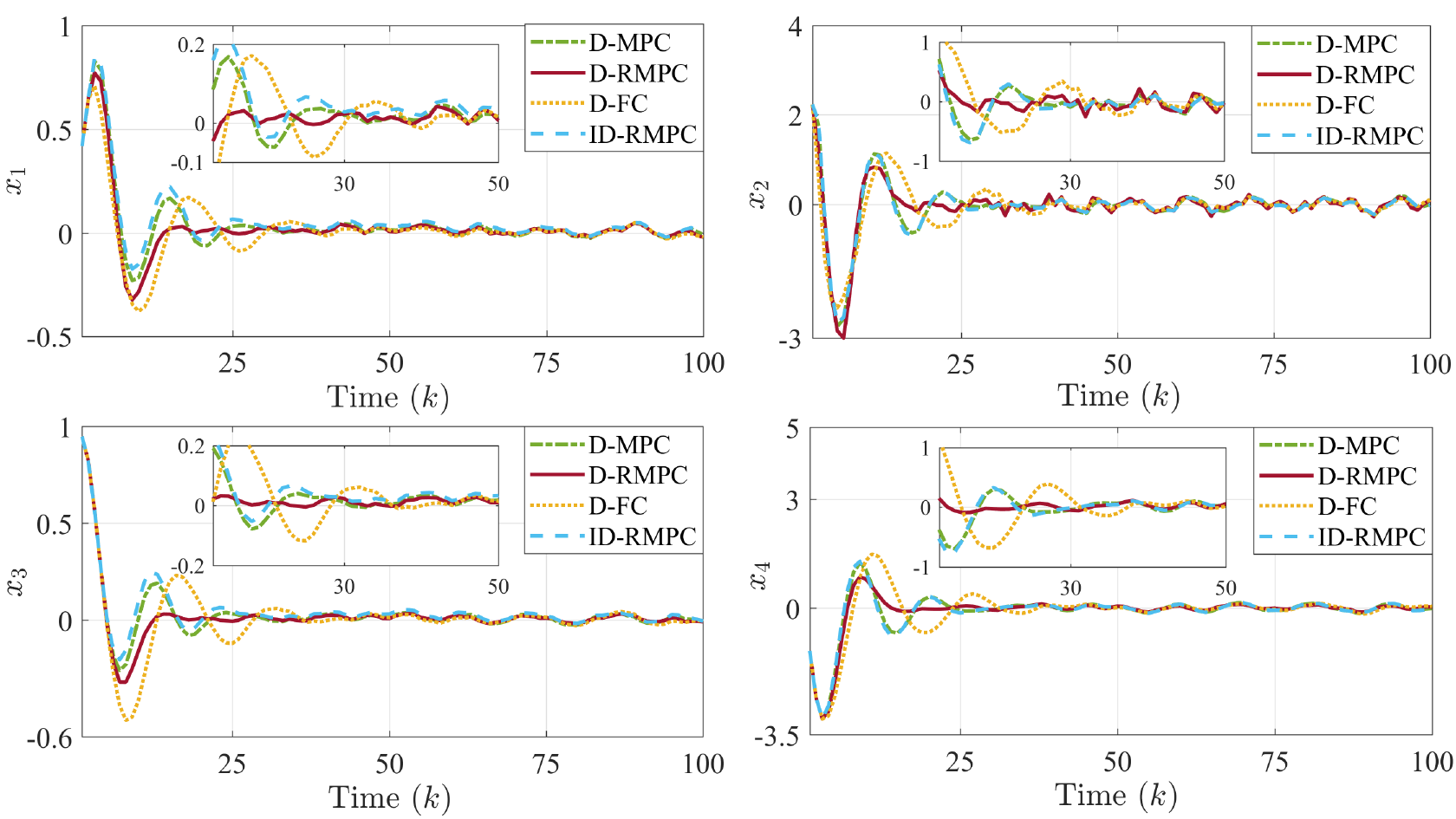} \caption{Closed-loop state trajectories under different methods with noisy online measurements (SNR = 62.2 dB)}
	\label{fig.x}
\end{figure}
\begin{figure}[h]
	\centering
	\includegraphics[width=0.4\textwidth]{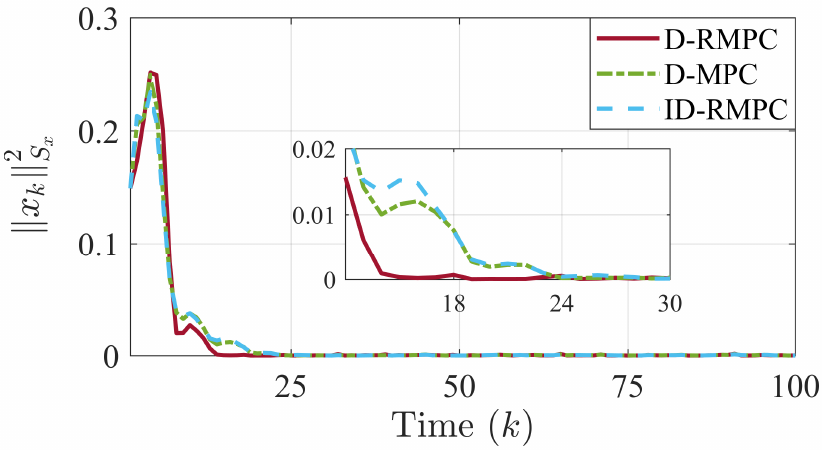} \caption{Closed-loop state constraint under different schemes.}
	\label{fig.xs1}
\end{figure}
\begin{figure}[h]
	\centering
	\includegraphics[width=0.48\textwidth]{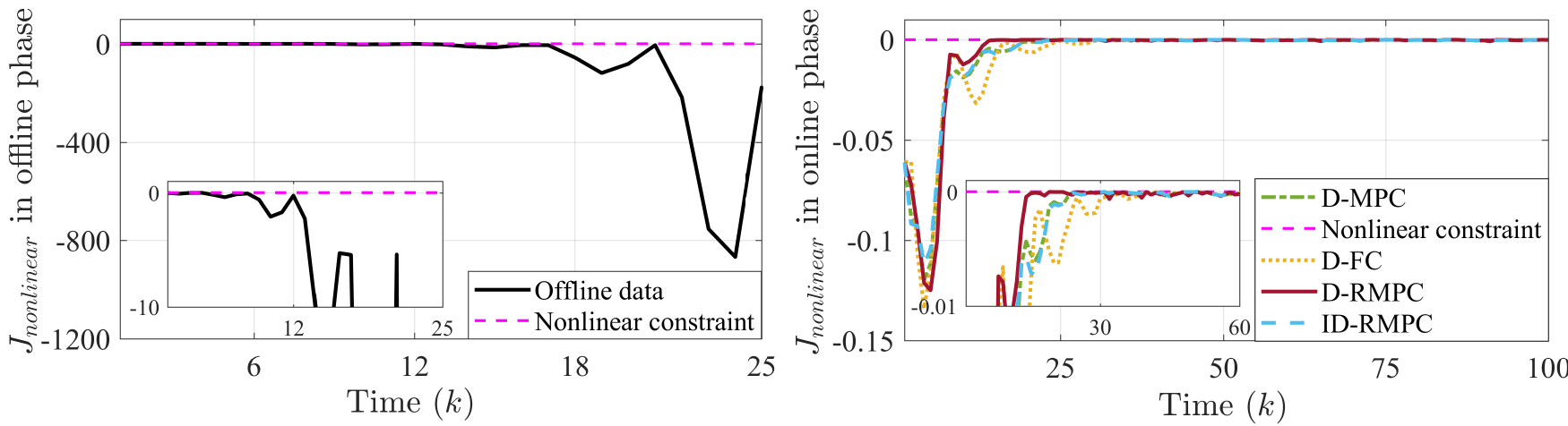} \caption{Nonlinear constraint during the offline collection phase and the online implementation phase}
	\label{fig.non}
\end{figure}

To provide a quantitative comparison, we perform 50 independent Monte Carlo experiments following the procedure described in Section~\ref{sec.sim1}. As summarized in Table~\ref{tab.two}, the proposed D-RMPC achieves the lowest closed-loop cost among all considered methods. Specifically, relative to D-RMPC, both D-MPC and ID-RMPC incur slightly higher closed-loop costs, with increases of approximately $7.91\%$ and $6.19\%$, respectively, whereas D-FC exhibits a substantially larger closed-loop cost, increasing by $60.71\%$. With respect to computation time, ID-RMPC achieves the shortest online runtime, since $(A_s, B_s)$ is identified offline via least-squares estimation and the subsequent online control law is based on the identified model. In contrast, both D-MPC and D-RMPC require solving an SDP involving the dataset $\mathcal{S}$ at each time step, leading to increased online computational time.
{\begin{table}[ht]
		\caption{Evaluation metrics of four methods}
		\centering	
		\begin{tabular}{cccc} 
			\toprule
			\multirow{1}{*}{Scheme}&{closed-loop cost}&{}&{computation time (s)} \\
			\midrule
			\multirow{1}{*}{D-RMPC} & {3.4953}&{}&{13.7220}  \\ 
			\multirow{1}{*}{D-MPC} & {3.7717}&{}&{14.0198}  \\ 
			\multirow{1}{*}{ID-RMPC} &3.7115&{}&{9.2132}  \\
			\multirow{1}{*}{D-FC} &5.6172&{}&{/}  \\
			\bottomrule
			\label{tab.two}	
		\end{tabular}
\end{table}}

Table~\ref{tab.two} reports the closed-loop cost and computation time, which correspond to the total performance cost and total simulation time accumulated over 50 Monte Carlo experiments. To provide further insight into the per-step computational burden associated with solving the SDP, as well as the evolution of the closed-loop performance, Fig.~\ref{fig.V1} shows the average and worst-case closed-loop cost at each time step. In addition, Fig.~\ref{fig.T1} depict the average and worst-case SDP computation time per step. As shown in Fig.~\ref{fig.V1}, D-RMPC attains the smallest closed-loop cost at the initial time step $k=1$. Under identical initial states, this indicates that D-RMPC generates smaller control inputs compared with the other methods. Moreover, both the average and worst-case closed-loop costs under D-RMPC converge to zero relatively quickly. From Fig.~\ref{fig.T1}, ID-RMPC requires significantly less computation time than D-RMPC and D-MPC. Furthermore, when disturbances are present during both the offline data collection and online implementation phases, D-RMPC exhibits slightly lower computational cost than D-MPC.

Indeed, regarding real-time applicability, D-MPC and D-RMPC exhibit relatively large initial computation times, and the per-step SDP solution time mainly ranges from 0.13 to 0.15s, slightly exceeding the adopted sampling period of 0.1s. Therefore, the current implementation should be regarded primarily as a simulation-based validation of the proposed methodology rather than a real-time execution. Nevertheless, the computational burden can be substantially reduced by incorporating warm-start strategies, tailored SDP solvers, and code-generation or low-level implementation techniques, which are expected to enable real-time operation. Investigating systematic approaches to reduce SDP solution time and achieve efficient real-time implementations represents a practically relevant direction for future research.
\begin{figure}[h]
	\centering
	\includegraphics[width=0.48\textwidth]{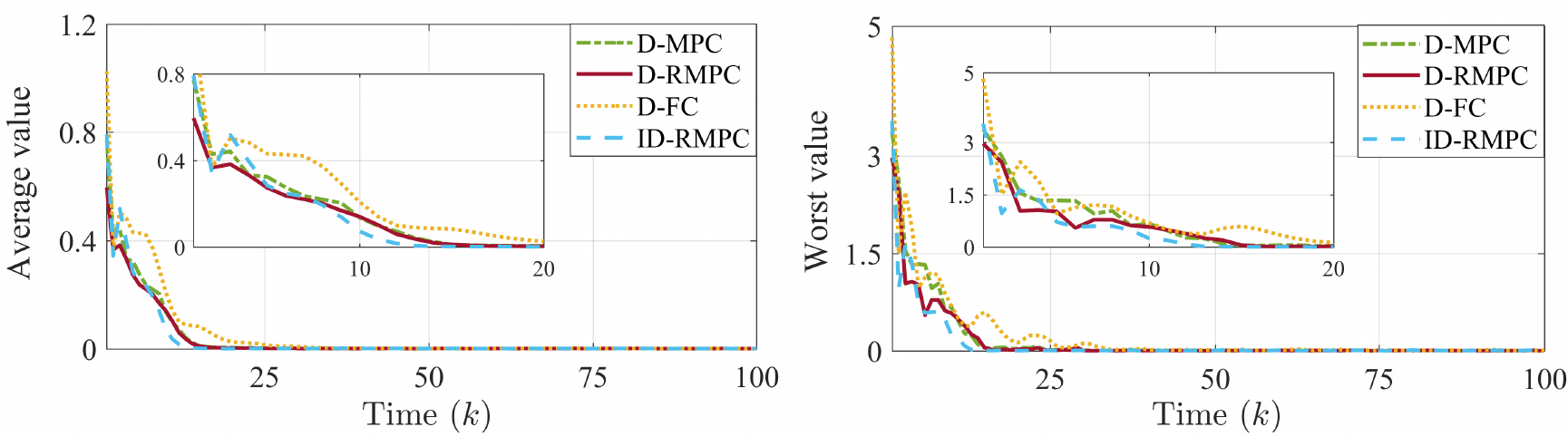} \caption{Average and worst-case closed-loop cost}
	\label{fig.V1}
\end{figure}
\begin{figure}[h]
	\centering
	\includegraphics[width=0.48\textwidth]{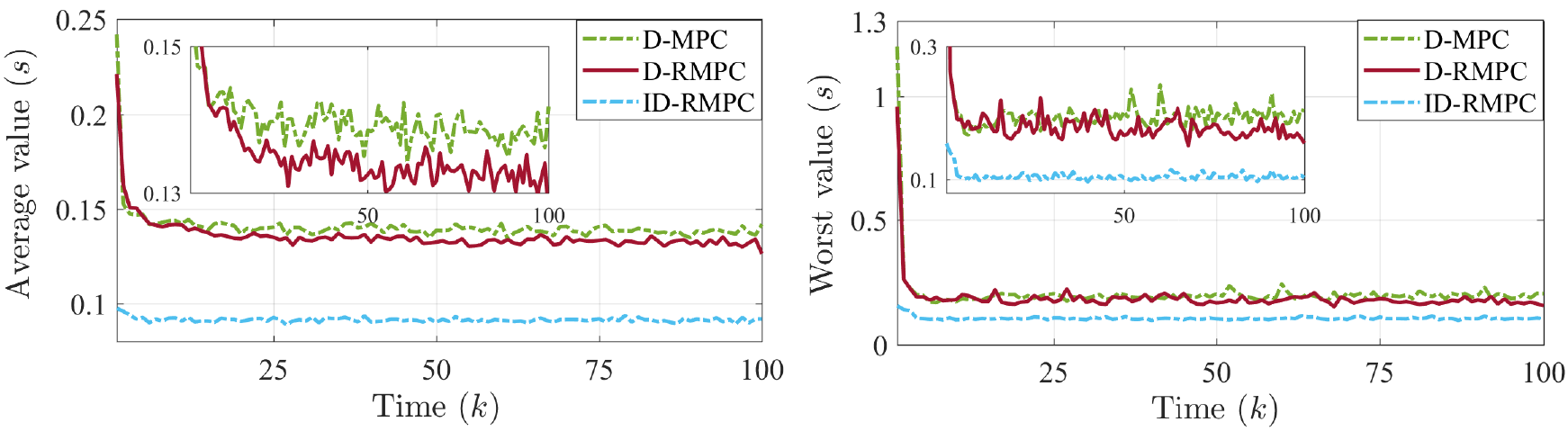} \caption{Average and worst-case SDP computation time}
	\label{fig.T1}
\end{figure}

\subsection{Conservatism and Feasibility of the Proposed Schemes }\label{sec.sim3}
In this section, we further examine the feasibility and conservatism of the proposed schemes using the flexible-joint robot system \eqref{eq.robot}. The analysis is conducted along three dimensions: i) the length of the historical data $N$, ii) the disturbance magnitude (and the associated signal-to-noise ratio, SNR), and iii) the choice of the disturbance bound $\varpi$, particularly when it is overly conservative or overly restrictive.  All parameters of the four schemes are selected consistently with those in Section~\ref{sec.sim2}.

We begin by investigating the effect of the historical data length. As discussed in Remark~9, a lossy S-procedure is employed to derive sufficient conditions ensuring that inequalities~\eqref{eq.augIAB} and \eqref{eq.augIAB_no} hold for all $(A,B)\in\Xi_k$, for $k=0,1,\ldots,N-1$. This relaxation is inherently conservative, and the resulting conservatism of the proposed framework is therefore closely tied to the amount of historical data used. To this end, the four control schemes are implemented using datasets of varying lengths.

Fig. \ref{fig.data} depicts the closed-loop cost and computation time for the four schemes as functions of $N$ over the range $[12,40]$. When the data length is insufficient (e.g., $N=10$), the SDPs in \eqref{sdp} and \eqref{sdpnoise} are infeasible at the initial time step. As $N$ increases to a moderate level, both the closed-loop cost and the computation time decrease, indicating improved performance and reduced conservatism. Notably, in the presence of disturbances during both the offline and online phases, D-RMPC consistently achieves the lowest closed-loop cost among the considered methods. However, as $N$ becomes excessively large, the closed-loop performance of the three data-driven schemes deteriorates and the computational burden increases. This degradation can be attributed to the fact that larger datasets may contain more divergent trajectories, resulting in a less accurate data-driven set-membership representation and, consequently, increased conservatism in the synthesized controller. In contrast, ID-RMPC, which first identifies the system matrices via a least-squares procedure and then designs the controller, exhibits more stable closed-loop performance and computation time with respect to the data length.
\begin{figure}[h]
	\centering
	\includegraphics[width=0.48\textwidth]{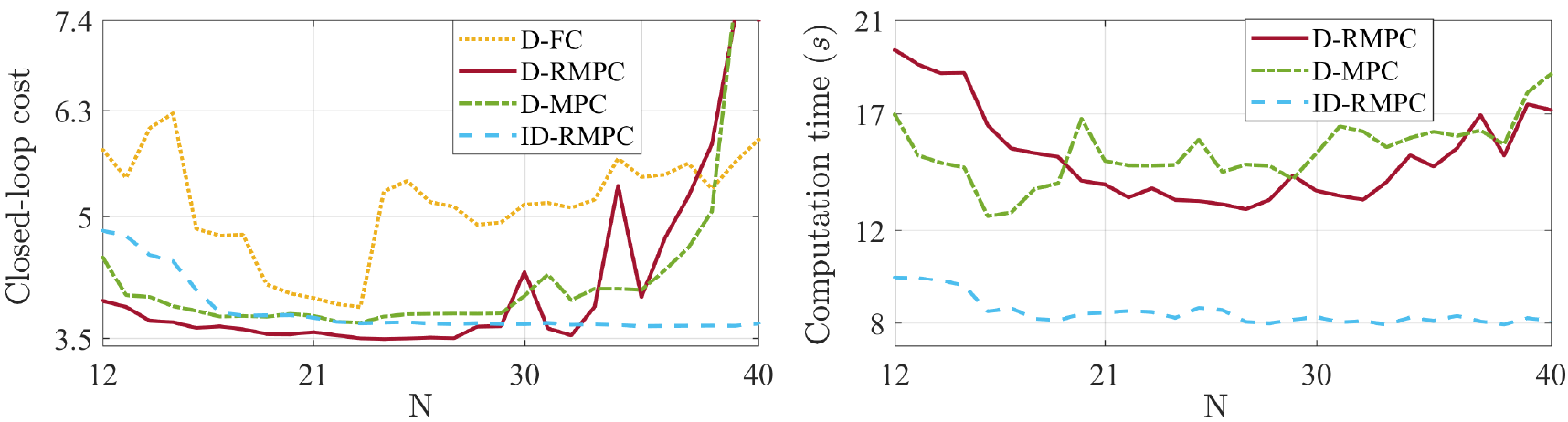} \caption{Closed-loop cost  and  computation time under four control schemes with different values of $N$}
	\label{fig.data}
\end{figure}

We next examine the effect of disturbance magnitude on the feasibility and closed-loop performance of the proposed schemes. In Fig.~\ref{fig.xno}, the disturbances are generated within the range $[-0.01,0.01]^4$. Both relatively large and small disturbance levels are considered in this section. Fig.~\ref{fig.high} depicts the closed-loop responses of the four control schemes when the disturbance lies in $[-0.1,0.1]^4$, corresponding to an SNR of 40.6~dB. Under this disturbance level, D-FC and D-MPC fail to stabilize the system and exhibit divergent behavior, whereas the robust schemes D-RMPC and ID-RMPC successfully stabilize the closed loop and ensure convergence to the RPI set. When the disturbance magnitude exceeds $[-0.1,0.1]^4$, the SDP problems \eqref{sdp} and \eqref{sdpnoise} become infeasible at the initial time step. Fig.~\ref{fig.low} reports the closed-loop performance when the disturbance is generated within $[-0.0001,0.0001]^4$, corresponding to an SNR of 135.6~dB. In this regime, all four control schemes robustly stabilize the system. Moreover, when the disturbance level is sufficiently small, the static controller D-FC exhibits a faster convergence rate than the remaining methods.
\begin{figure*}[h]
	\centering
	\includegraphics[width=0.9\textwidth]{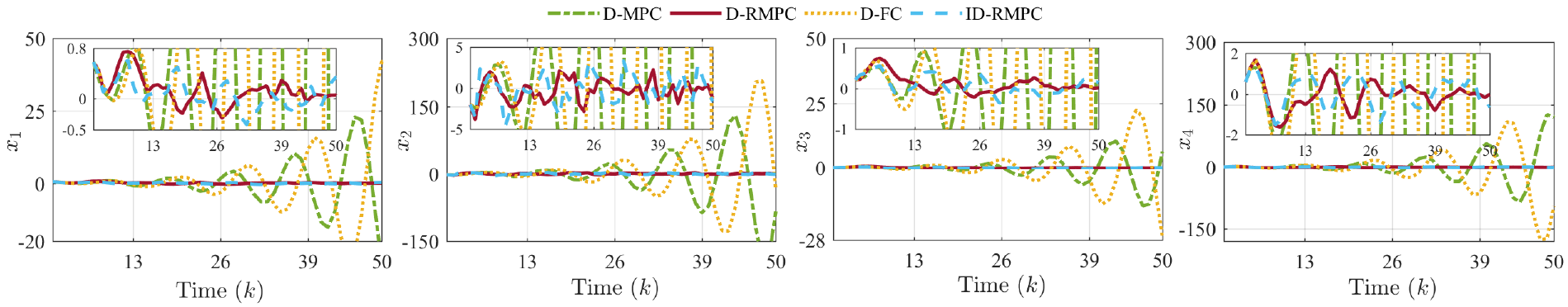} \caption{Closed-loop state trajectories under four control schemes with $d_k \in [-0.1,0.1]^4$ (SNR = 40.6 dB)}
	\label{fig.high}
\end{figure*}
\begin{figure*}[h]
	\centering
	\includegraphics[width=0.9\textwidth]{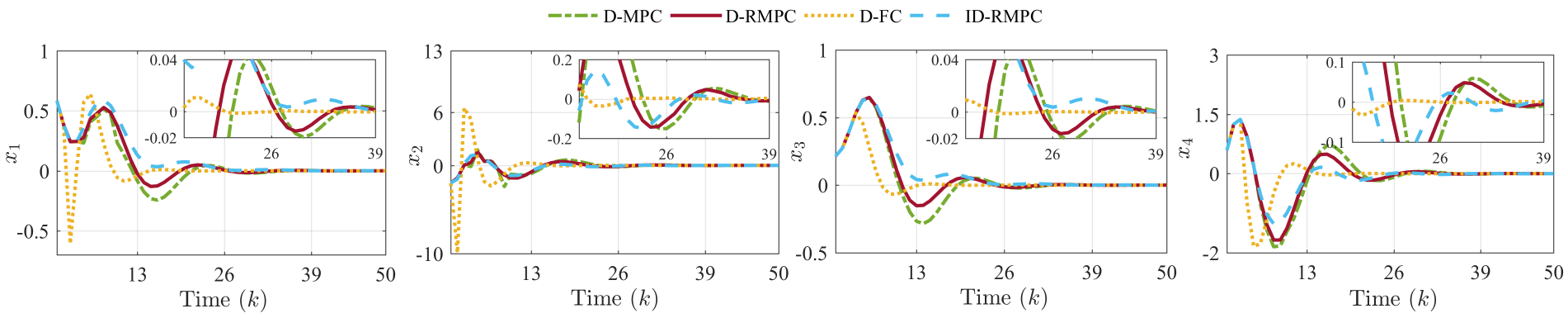} \caption{Closed-loop state trajectories under four control schemes with $d_k \in [-0.0001,0.0001]^4$ (SNR = 135.6 dB)}
	\label{fig.low}
\end{figure*}

We further examine the sensitivity of the proposed methods to inaccuracies in the disturbance bound $\varpi$. The control performance of D-MPC without online disturbances and that of D-RMPC under online disturbances are considered.
The system and parameter settings are chosen consistently with those in Section~\ref{sec.sim2}. The disturbance is generated within the set $[-0.01,0.01]^4$, for which the nominal bound satisfies $\varpi^2 = 4 \times 10^{-4}$. In addition, we consider two mismatched cases in which the assumed bound is either larger or smaller than the true disturbance magnitude, namely $\varpi^2 = 2 \times 10^{-3}$ and $\varpi^2 = 3 \times 10^{-5}$, respectively, as illustrated in Figs.~\ref{fig.bound1}–\ref{fig.bound}.

When $\varpi$ is underestimated, the closed-loop performance deteriorates and the state trajectories exhibit pronounced oscillations, particularly in the presence of disturbances during both the offline and online phases. This behavior arises because an underestimated $\varpi$ fails to cover all disturbance realizations, which compromises the construction of the data-driven set-membership representation. Moreover, the resulting RPI set is smaller than the true RPI set, ultimately leading to degraded closed-loop performance. Conversely, when $\varpi$ is overestimated, the closed-loop performance is also inferior to the nominal case with $\varpi^2 = 4 \times 10^{-4}$. The resulting RPI set becomes overly conservative, and the system state may initially lie entirely within the RPI set, causing the controller to rely on a static state-feedback gain and yielding degraded performance.
\begin{figure*}[h]
	\centering
	\includegraphics[width=0.9\textwidth]{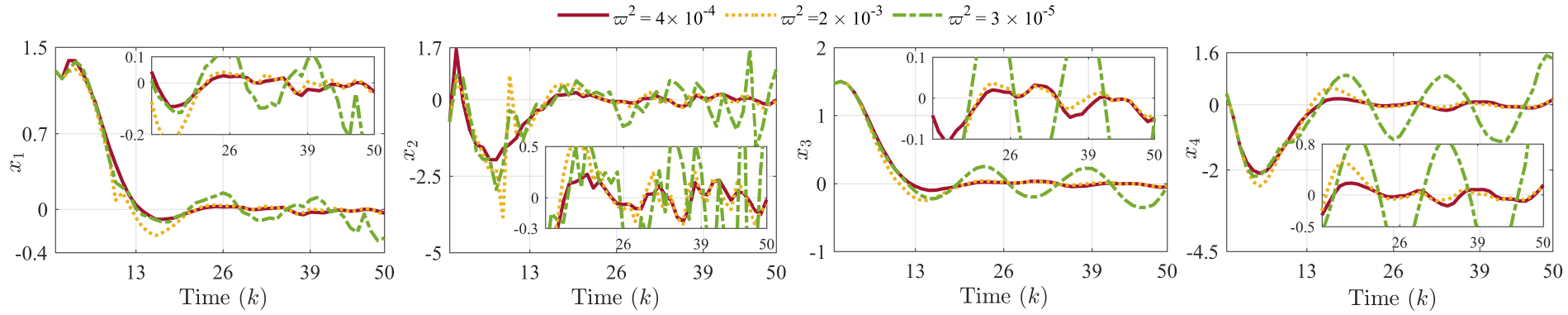} \caption{Closed-loop state trajectories under D-RMPC with different values of $\varpi$}
	\label{fig.bound1}
\end{figure*}
\begin{figure*}[h]
	\centering
	\includegraphics[width=0.9\textwidth]{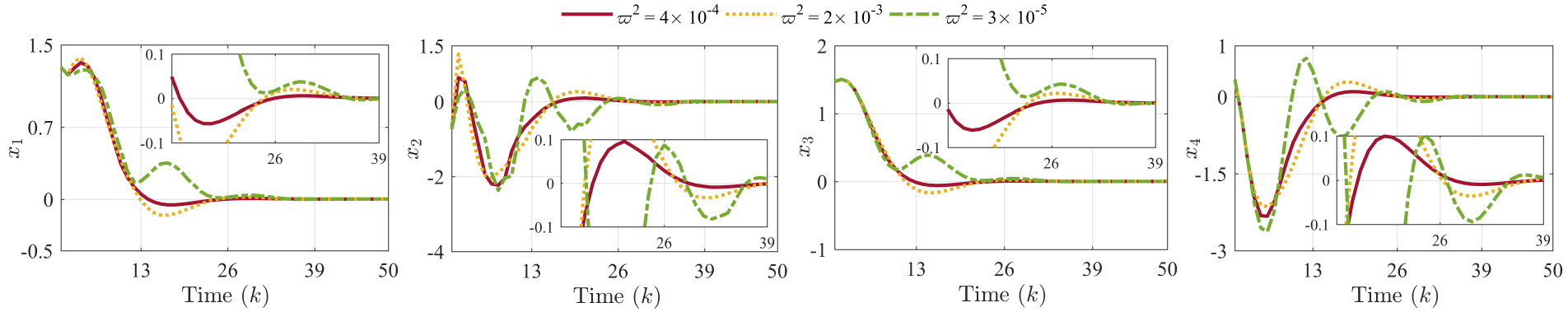} \caption{Closed-loop state trajectories under D-MPC with different values of $\varpi$}
	\label{fig.bound}
\end{figure*}

\subsection{The Impact of the Misspecified Dictionary $Z(x_k)$}\label{sec.sim4}
This section investigates the impact of an incorrect selection of the function dictionary on the performance of the proposed schemes. Since the designs of D-MPC and D-RMPC share a common structure, we focus, for brevity, on D-MPC operating in the absence of online disturbances. The flexible-joint robot system and all algorithmic parameters are chosen consistently with those in Section~\ref{sec.sim1}, where the nonlinear component of the nominal function dictionary $Z(x_k)$ is chosen as $Q(x_k) = \sin(x_{3,k})$, while the linear component is set to $x_k$. To evaluate the sensitivity of the proposed approach to modeling mismatch, we consider several misspecified function dictionaries, namely $Q(x_k)=\cos(x_{3,k})$, $Q(x_k)=\sin(x_{2,k})$, and $Q(x_k)=3\sin(x_{3,k})$. The corresponding simulation results are reported in Fig.~\ref{fig.di}.

These results demonstrate how deviations from the nominal dictionary affect closed-loop performance and feasibility. Specifically, Fig.~\ref{fig.di} presents the closed-loop state trajectories of D-MPC over a simulation horizon of 50 time steps. When $Q(x_k)=\cos(x_{3,k})$, certain state components fail to converge asymptotically to the origin. For the other two misspecified choices of $Q(x_k)$, closed-loop stability is preserved; however, the resulting performance is significantly degraded compared with the nominal case.
\begin{figure*}[h]
	\centering
	\includegraphics[width=0.9\textwidth]{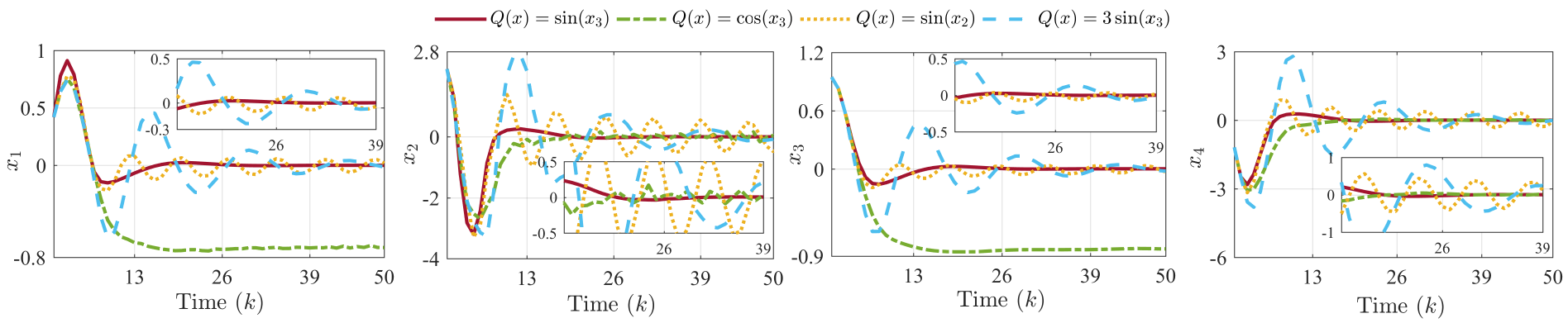} \caption{Closed-loop state trajectories under D-MPC with different $Q(x_k)$}
	\label{fig.di}
\end{figure*}

\section{Conclusions}\label{sec.con} 
We presented a data-driven robust min–max MPC framework for unknown nonlinear systems with process disturbances. By representing the nonlinear vector field via a known basis (yielding a lifted linear form) and using offline noisy data to define a set-membership characterization of system matrices, we proposed a data-driven MPC problem with the ellipsoidal state constraint. Two scenarios were investigated: in one, the online system is not affected by process disturbances, while in the other, it is subject to such disturbances. In each case, an SDP was formulated to derive a state-feedback control gain. We proved that the resulting controllers guarantee recursive feasibility, closed-loop stability, and constraint satisfaction. Numerical examples demonstrated the effectiveness of the approach in handling nonlinear dynamics and disturbances compared with other data-driven controllers. Future work will explore extensions to distributed systems and high-order fully actuated systems, as well as reductions in SDP computation time.

\bibliographystyle{plainnat}
\bibliography{reference}
\end{document}